\documentclass[11pt,article,reqno]{amsart}
\usepackage{amsmath, amsthm, amsfonts}
\usepackage{amssymb,url}
\usepackage{mathtools}
\usepackage{hyperref}
\usepackage{amsaddr}
\usepackage{array}
\usepackage[top=29mm,bottom=29mm,left=30mm,right=30mm]{geometry}
\usepackage{multirow}

\usepackage{mathrsfs}
\usepackage[final]{graphicx}
\usepackage{ulem}
\usepackage{color}
\usepackage{mathtools}
\usepackage{graphicx}

\makeatletter
\def\section{\@startsection{section}{1}%
\z@{1\linespacing\@plus\linespacing}{1\linespacing}%
{\bf\centering}}
\def\subsection{\@startsection{subsection}{0}%
\z@{\linespacing\@plus\linespacing}{\linespacing}%
{\bf}}
\def\subsubsection{\@startsection{subsubsection}{0}%
\z@{\linespacing\@plus\linespacing}{\linespacing}%
{\bf}}
\makeatother

\makeatletter
\@addtoreset{equation}{section}
\makeatother

\usepackage{adjustbox}
\newtheorem{theorem}{Theorem}[section]
\newtheorem{corollary}[theorem]{Corollary}

\theoremstyle{definition}

\newtheorem{assumption}[theorem]{Assumption}

\begin{document}
\title
{Fractional binomial regression model for count data with excess zeros }
\author{Jeonghwa Lee \and Chloe Breece }

\address{ University of North Carolina Wilmington, USA}

\thanks{  Email address:  J. Lee (corresponding author): leejb@uncw.edu, C. Breece: ceb1386@uncw.edu. 
}

\begin{abstract}
 This paper proposes a new generalized linear model with the fractional binomial distribution. 
 Zero-inflated Poisson/negative binomial distributions are used for count data with many zeros. To analyze the association of such a count variable with covariates, zero-inflated Poisson/negative binomial regression models are widely used. In this work, we develop a regression model with the fractional binomial distribution that can serve as an additional tool for modeling the count response variable with covariates. The consistency of maximum likelihood estimators of the proposed model is investigated theoretically and empirically with simulations. The practicality of the proposed model is examined through data analysis. The results show that our model is as versatile as or more versatile than the existing zero-inflated models, and especially, it has a better fit with left-skewed discrete data than other models. However, the proposed model faces computational obstacles and will require more work in the future to implement this model on various count data with excess zeros. 
\end{abstract}
\maketitle
\baselineskip 0.5 cm

\bigskip \medskip



\section{Introduction}
Count data with excess zeros are often found in many areas, including public health, epidemiology, ecology, finance, and quality control, to name a few.
 In the past, several zero-adjusted discrete models have been developed for such a data set \cite{Gup}. When there is additional information from covariates and one wishes to analyze the association of covariates with a count response variable with many zeros, then the zero-inflated Poisson/negative binomial regression model can be used \cite{Gre, Hei, Lam, Mul}. In \cite{Hal}, the zero-inflated binomial regression model was developed for the upper bounded count response.  
These zero-inflated models are mixed models in which structural zeros are modeled separately and added to a regular count model that generates sampling zeros and other count values.  

In this work, we propose a new model that can serve as an additional tool for a regression model for count data with excess zeros. This model, which we will call a fractional binomial (FB) regression model, is developed from the fractional binomial distribution first introduced in \cite{Lee}. The FB regression model takes a different approach than the aforementioned zero-inflated models to modeling count data with excess zeros. It does not treat sampling zeros and structural zeros separately, yet it is as flexible as or, at times, more flexible than the conventional zero-inflated models so that it can fit count data of various shapes.   The FB regression is based on the fractional binomial distribution which is derived from correlated Bernoulli variables. It is this dependency among the Bernoulli variables that results in overdispersion and zero-inflation in the discrete model, and it also results in more flexible shapes of the distribution than the conventional discrete distributions. 

The fractional binomial distribution is defined by the sum in a stationary binary sequence, called a generalized Bernoulli process (GBP), whose covariance function decays by power law \cite{Lee}. The applications of GBP and the fractional binomial distribution are found in seismology, economics, and horticulture \cite{Lee3, Lee2}. The mean of the fractional binomial distribution is the same as that of the regular binomial distribution since it is defined from a stationary sequence, but the variance of the fractional distribution is larger than that of the regular binomial distribution due to the dependence among the Bernoulli variables. Especially, under certain conditions, its variance is proportional to the length of the binary sequence to a fractional power.   The shape of the fractional binomial distribution varies with parameters that affect the center, spread, skewness, and zero-inflation of the distribution. 
 Using a link function, one can incorporate covariates to connect them to the shape of the distribution of the response variable. 

We investigate the consistency of the maximum likelihood (ML) estimator of the FB regression model theoretically and empirically with simulations. We also examine the applicability of the FB model through data analysis and assess the strengths and weaknesses of the proposed model. Our results show that when data size is small with a few covariates (approximately less than five covariates) and a small or moderate bound on the maximum response (approximately less than 20), the FB model can perform as well as or better than the existing zero-inflated models. Specifically, the FB model is found to have a better fit than other models when the count response has a left-skewed distribution.  However, the FB model has computational challenges; it suffers from computational instability in finding ML estimators and is not suitable for large datasets. 

In Section 2, the zero-inflated regression models are reviewed. In Section 3, the FB regression model is introduced and the consistency of the ML estimators is investigated, followed by simulations to examine the performance of the estimators in Section 4. In Section 5, the zero-inflated regression models and the FB model are fitted to two datasets, their fits are compared by AIC, and the goodness of fit of the models is checked through randomized quantile residuals.  The conclusion and discussion are followed in Section 6.

\section{Zero-inflated regression models}
Zero-inflated (ZI) regression models are mixed models of regular regression models for discrete data with a nondegenerate probability distribution at zero. 
Let $y_i\in \mathbb{N}\cup\{0\}$ be a response variable and $(x_{1i}, x_{2i},\cdots, x_{ki} )\in \mathbb{R}^k $ be covariates for
$i=1,2,\cdots,n$. 
Let $ {\bf x}_i=(1,x_{1i}, x_{2i},\cdots, x_{ki} ),$
$ {\boldsymbol \gamma}=(\gamma_0,\gamma_1,\gamma_2,\cdots,\gamma_k) , $ $ {\boldsymbol \tau}=(\tau_0,\tau_1,\tau_2,\cdots,\tau_k),$ and 
$ {\boldsymbol \beta}=(\beta_0,\beta_1,\beta_2,\cdots,\beta_k).$    The probability distributions of the zero-inflated binomial/Poisson/negative binomial (ZIB, ZIP, ZINB) regression models are expressed as
\[ P(y_i|{\bf x}_i, \Theta )= \begin{dcases}
\pi_i+(1-\pi_i)g(0| {\bf x}_i) &\text{if } y_i=0,\\ 
(1-\pi_i)g(y_i| {\bf x}_i )&\text{if } y_i\in \mathbb{N},
\end{dcases}
\]
where
\[ \pi_i=\frac{1}{1+\exp(-{\bf x}_i'{\boldsymbol \gamma} )}, \]
and
\[ g(y_i| {\bf x}_i )= \begin{dcases} {N \choose y_i}p_i^{y_i} (1-p_i)^{N-y_i}&\text{for ZIB regression model,}\\ \frac{e^{-\mu_i} \mu_i^{y_i}}{y_i!}&\text{for ZIP regression model,}
\\  \frac{ \theta^{\theta} \mu_i^{y_i} \Gamma(\theta+y_i)}{\Gamma(\theta) \Gamma(y_i+1)(\theta+\mu_i)^{\theta+y_i }}
 &\text{for ZINB regression model,}
\end{dcases} \]
where
\[ p_i=\frac{1}{1+\exp(-{\bf x}_i'{\boldsymbol \tau} )} \text{ and  }  \mu_i=\exp({\bf x}_i'{\boldsymbol \beta} ).\] In \cite{Hal}, the ZIB model was developed for count response that has an upper bound. But for the purpose of describing data with a probability model without making a statistical inference beyond the range of dataset, all zero-inflated models can be used regardless of whether the response variable has a theoretical upper bound or not. Therefore, in this paper, we use all zero-inflated models and develop the FB model for both bounded and unbounded count data. 
Here, we denote by $N$ the maximum value of the response in a dataset or the upper bound of the response variable if it has one. The ZIB we use in this paper is a simpler version of that in \cite{Hal}, as we have a constant upper bound $N$ while in \cite{Hal} the upper bound of the response varies and is denoted by $n_i$.

Note that $g$ is the probability distribution of standard count regression models where $\mu_i$ is the conditional mean given the covariates in the case of Poisson and negative binomial regression models, and  $\pi_i$ gives the additional probability at zero. In the Poisson regression model, given the values of covariates, the conditional mean  and the conditional variance   are the same as $\mu_i.$
In the negative binomial regression model, the conditional variance is $\mu_i+\mu_i^2/\theta,$ therefore, the negative binomial model can serve for overdispersed count data. We can also generalize the dispersion parameter $\theta$ in ZINB  and replace $\theta$ by
\[\theta_i=\exp({\bf x}_i'{\boldsymbol \alpha} ),\] where $ \boldsymbol \alpha=(\alpha_0,\alpha_1,\cdots,\alpha_k)$, and we call it ZINB-2.

Let $\Theta$ denote the set of all parameters in the corresponding model, i.e.,  $\Theta=({\boldsymbol \beta},{\boldsymbol \tau})$ for ZIB,   $\Theta=({\boldsymbol \beta},{\boldsymbol \gamma})$ for ZIP,  $\Theta=({\boldsymbol \beta},{\boldsymbol \gamma},\theta)$ for ZINB,  and $\Theta=({\boldsymbol \beta},{\boldsymbol \gamma},\alpha)$ for ZINB-2.
The parameters of these models are estimated by the maximum likelihood estimation (MLE),
\[ \hat{\Theta}= \operatorname*{arg\,max}_{\Theta \in \Omega} L( \Theta | (y_i, {\bf x}_i), i=1,2,\cdots,n )=\operatorname*{arg\,max}_{\Theta \in \Omega} \prod_{i=1}^n  P(y_i| {\bf x}_i, \Theta)\]
where $\Omega=\mathbb{R}^{dim(\Theta)}.$
The asymptotic properties of the MLE of
zero-inflated models were studied in \cite{Cza, Sta}.
\section{Fractional binomial regression model}

The fractional binomial (FB)
distribution, introduced in \cite{Lee}, is defined by the cumulative sum in a sequence of dependent Bernoulli trials, called a generalized Bernoulli process (GBP).
 The GBP, $\{\xi_i, i\in \mathbb{N}\}$, is a sequence of stationary binary variables such that for any $i\in \mathbb{N},$
 \begin{align*}
 &P(\xi_i=1)=p, P(\xi_i=0)=1-p ,\end{align*} and for any 
 $ 0<i_0<i_1<i_2<\cdots<i_n,$
\begin{equation}
       P(\xi_{i_0}=1, \xi_{i_1}=1, \cdots, \xi_{i_n}=1)=
p\prod_{j=1}^{n}(p+c|i_j-i_{j-1}|^{2H-2}). \end{equation} 
More generally, for any disjoint sets $A,B \subset \mathbb{N},$ the joint probability distribution in GBP is defined by the inclusion-exclusion principle,
\begin{align}
    &P( \cap_{i'\in B }\{\xi_{i'}=0\}  \cap_{i\in A }   \{\xi_{i}=1\} )=\sum_{k=0}^{|B|}\sum_{\substack{B'\subset B\\ |B'|=k}} (-1)^{k}P(\cap_{i\in B'\cup A}\{\xi_{i}=1\} ),
\end{align}
and
\begin{align}
    P(\cap_{i'\in B }\{\xi_{i'}=0\})&=1+\sum_{k=1}^{|B|}\sum_{\substack{B'\subset B\\ |B'|=k}} (-1)^{k}P(\cap_{i\in B'}\{\xi_{i}=1\} ),
\end{align}
with parameters $(p,H,c)$ that satisfy the following assumption.
\begin{assumption}
 $p,H \in (0,1),$ and \[
    0\leq c<\min\{1-p, \frac{1}{2} ( -2p+2^{2H-2} +\sqrt{4p-p2^{2H}+2^{4H-4}})\}.
\] \label{Assumption 2.1}
\end{assumption}

The GBP is a stationary binary sequence with a covariance function
\[ {Cov}(\xi_i, \xi_j)=pc|i-j|^{2H-2}, i\neq j.\] When $H\in(0.5,1),$  the GBP possesses long-range dependence, since $\sum_{i=1}^{\infty} {Cov}(\xi_1, \xi_i)=\infty.$ Fractional binomial random variable, denoted by $B_N(p,H,c)$, is defined as the sum of the first $N$ variables in the GBP. Its mean is $Np$, and the variance is

\begin{align*}
 E((B_N-Np)^2)&\sim\begin{dcases}
    b_1 N  &\text{ if $H\in(0, .5),$}\\
   b_2 N\ln{N}  &\text{ if $H = .5,$} \\
    b_3 N^{2H}  &\text{ if $H\in(.5, 1),$}
\end{dcases}  \end{align*}
where $ b_1=  p(1-p)+2pc/(1-2H), b_2=  2pc, $ and $ b_3= {pc  }/{(H(2H-1))}.$  If $H\in (0.5,1),$ its variance is asymptotically proportional to $N$ to a fractional power.
 When $c=0, $ $B_N(p,H,0)$ becomes the regular binomial random variable whose parameters are $N,p.$ The probability mass function (pmf) of the fractional binomial distribution does not have a closed-form expression, however, it can be computed iteratively through equations (3.1-3.3). 

 Figure 1 shows the pmfs with various sets of parameters. R package ``frbinom" was used to generate the pmfs with $N=30$.  
\begin{figure}[b]
    \centering
    \includegraphics[width=\linewidth]{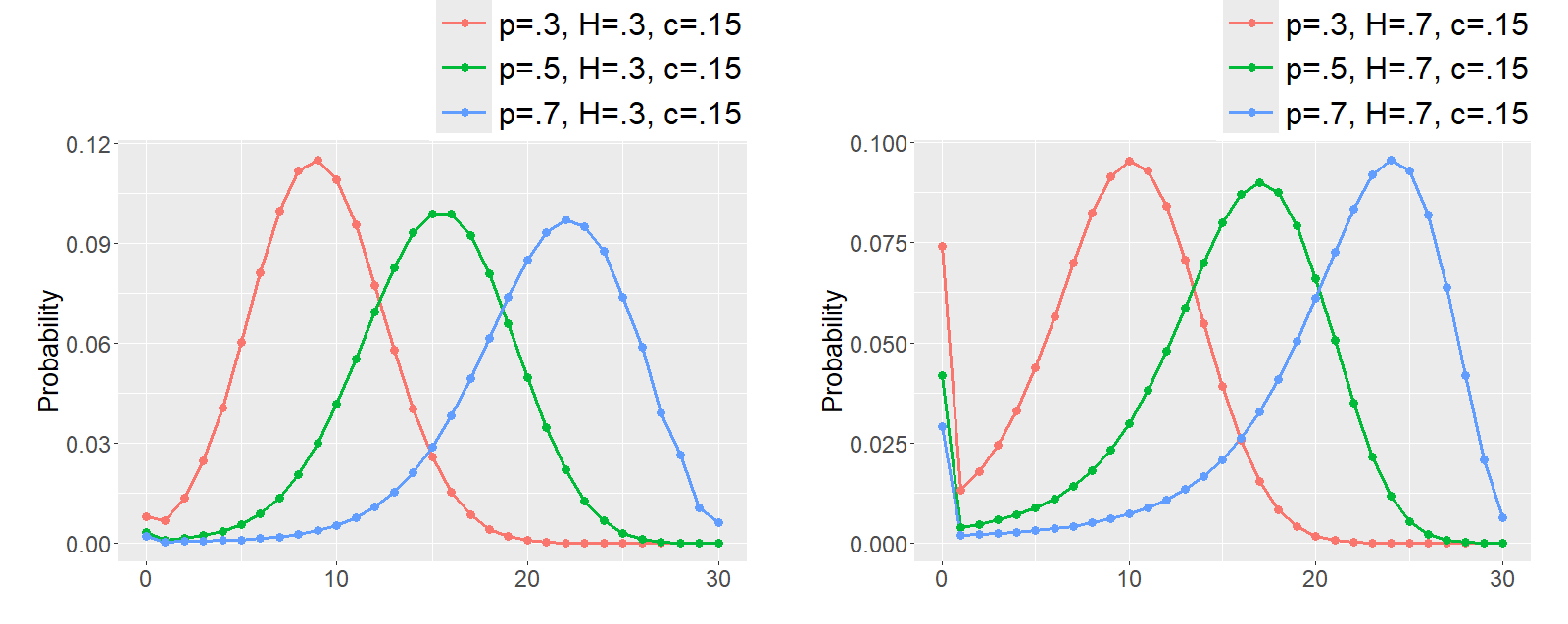}
    \includegraphics[width=\linewidth]{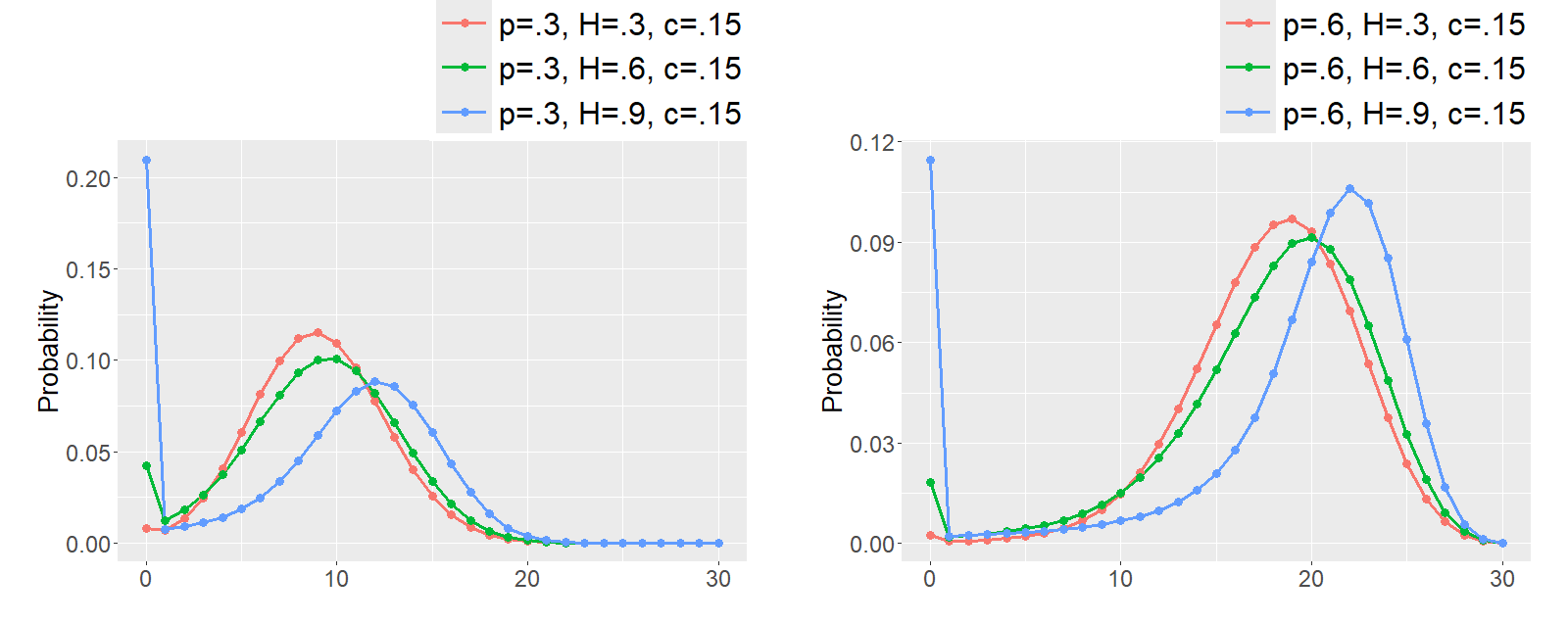}
    \includegraphics[width=\linewidth]{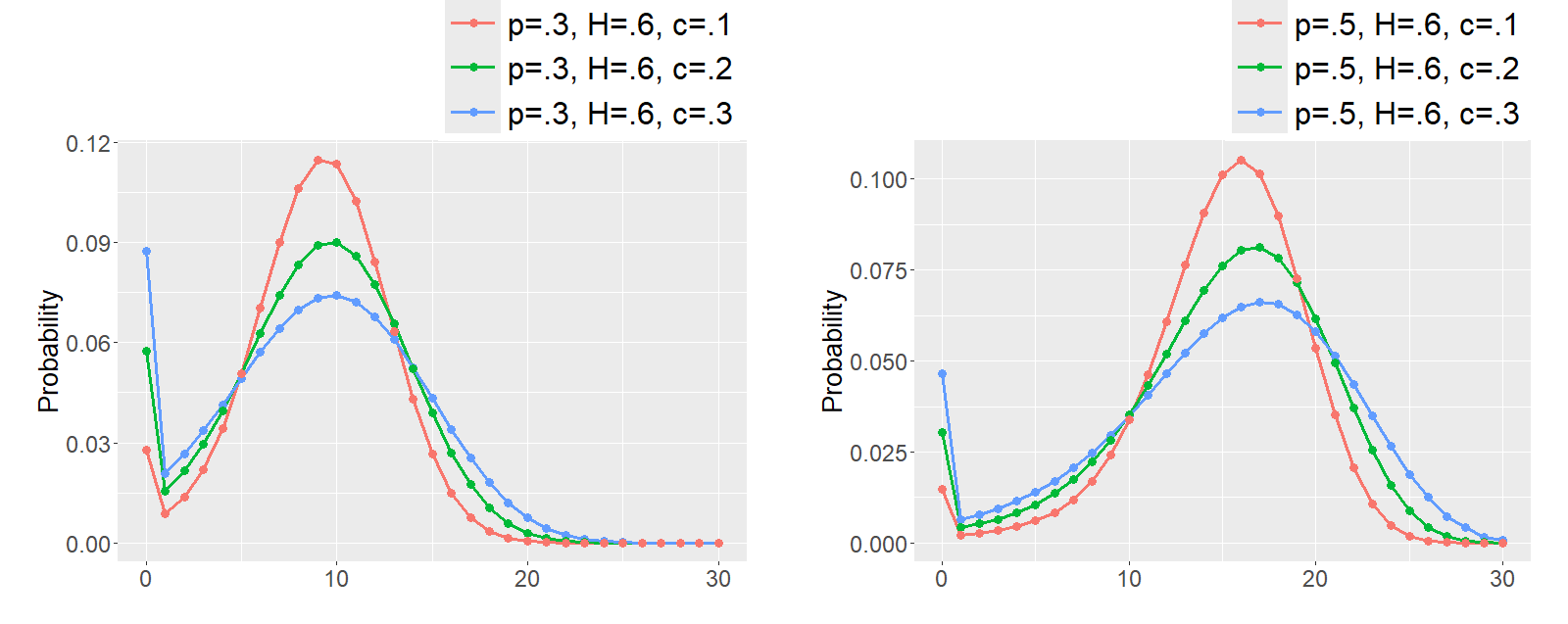}
    \caption{Probability mass function of fractional binomial distribution with changing parameter of p(top), H(middle), c(bottom) }
    \label{fig:enter-label}
\end{figure}
In Figure 1, the plots in the first row show the effect of varying $p$ when other parameters are fixed, which changes the location of the distribution. 
The plots in the middle show that changes in  $H$ affect the zero-inflation and skewness of the distribution. The larger $H$, the greater the zero inflation and skewness of the distribution. 
Note also that the peak at zero starts to appear when $H>.5$. The plots in the bottom row show that the parameter $c$ affects the spread of the distribution. The larger $c$, the larger the spread and the higher the zero inflation.    

 By connecting the parameters with covariates through a logit link, we can define the fractional binomial regression model as follows.
Let  ${\boldsymbol \psi}=(\psi_0,\psi_1,\cdots,\psi_k),$ ${\boldsymbol \eta}=(\eta_0,\eta_1,\cdots,\eta_k),$ and $ {\boldsymbol u}=(u_0, u_1,\cdots,u_k  ) .$ Given covariates ${\bf x}_i \in \mathbb{R}^k, i=1,2,\cdots,n$, the response variable follows the fractional binomial distribution,
\[y_i|{\bf x}_i \sim B_N(p_i, H_i, c_i),\] with
\begin{align}
    p_i=& \frac{1}{1+\exp(-{\bf x}_i' {\boldsymbol \psi} )},\\
    H_i=&\frac{1}{1+\exp(-{\bf x}_i' {\boldsymbol \eta} )}, \\
    c_i= &\frac{1}{1+\exp(-{\bf x}_i' {\boldsymbol u} )}, \nonumber
\end{align}
 and  $({\boldsymbol \psi}, {\boldsymbol \eta}, {\boldsymbol u})\in \Omega_{\{{\bf x}_1, {\bf x}_2,\cdots, {\bf x}_n \}},$ where  $\Omega_{\{{\bf x}_1, {\bf x}_2,\cdots, {\bf x}_n \}}$
is a set of all $({\boldsymbol \psi}, {\boldsymbol \eta}, {\boldsymbol u})\in\mathbb{R}^{3(k+1)}$ such that $(p_i,H_i, c_i)$  satisfies Assumption 3.1  for all $i=1,2,\cdots ,n.$
Like ZIB, the value of $N$ is set as the maximum value that the response variable can take if it has an upper bound, or the maximum response in a dataset. Finding the parameter space $\Omega_{\{{\bf x}_1, {\bf x}_2,\cdots, {\bf x}_n \}}$ of $({\boldsymbol \psi}, {\boldsymbol \eta}, {\boldsymbol u})$  is not an easy task, and one can expect that it might have a complicated form. Instead, we can reparameterize  $(p, H, c)$ in fractional binomial distribution by $(p, H, c^{\circ}),$ where \begin{align*}
 c= c^{\circ}\times  {\min\{1-p, \frac{1}{2} ( -2p+2^{2H-2} +\sqrt{4p-p2^{2H}+2^{4H-4}})\}} ,\end{align*} so that Assumption 3.1 is satisfied with $ (p_i, H_i, c_i^{\circ}) \in (0,1)^3,$ and we let
 \begin{align}
 c_i^{\circ}=\frac{1}{1+\exp(-  {\bf x}_i'\nu)} .\end{align}
 This results in a simple parameter space 
  $\Theta=({\boldsymbol \psi}, {\boldsymbol \eta}, {\boldsymbol \nu}) \in \Omega=\mathbb{R}^{3(k+1)}.$   The parameters $\Theta=({\boldsymbol \psi}, {\boldsymbol \eta}, {\boldsymbol \nu})$   are estimated by MLE,
\[ \hat{\Theta}= \operatorname*{arg\,max}_{\Theta \in \Omega} L( \Theta | (y_i, {\bf x}_i), i=1,2,\cdots,n )=\operatorname*{arg\,max}_{\Theta \in \Omega} \prod_{i=1}^n  P(y_i| {\bf x}_i, \Theta)\]
where $   P(y_i|  {\bf x}_i, \Theta)=P(B_N(p_i,H_i, c_i^{\circ})=y_i),$ and $p_i, H_i, c_i^{\circ}$ are (3.4), (3.5), and (3.6), respectively. As the pmf of the fractional binomial distribution does not have an explicit expression, MLE is obtained through numerical optimization with the  parameter space $\Theta\in \Omega=\mathbb{R}^{3(k+1)}.$ 
As the reparameterization simplifies the parameter space of the FB regression model, we will use this reparameterized model for the rest of the paper.


For the ML estimator to be consistent in the FB regression model, we need the following assumptions on the covariates and parameters.

\begin{assumption}
{\it i)} Covariates are i.i.d. random vectors and bounded, i.e., $ {\bf X}_i \stackrel{i.i.d.}{\sim}g $ and the support set of the distribution $g$ is bounded, ${\bf X }_i=(x_{i1},x_{i2},\cdots, x_{ik}) \in \mathcal{X}$ for all $i=1,2,\cdots$, where  $ \mathcal{X}$ is a compact set in $\mathbb{R}^k$. 
\\{\it ii)} Covariates are linearly independent and not constant, i.e.,   there is no set of constants $(a_1,a_2,\cdots, a_k)\in\mathbb{R}^{k}/\{{\bf 0}\}$ such that $ var( a_1x_{i1}+a_2 x_{i2}+\cdots +a_kx_{ik})=0.$   
\\
{\it iii)} The true parameter value $\Theta_0$ lies in a known compact set, i.e., $\Theta_0 \in  \mathcal{C}$ where $\mathcal{C}$ is a compact set in $\mathbb{R}^{3(k+1)}.$

\end{assumption}

Here, we assume $N$ is known, but if $N$ is not unknown but lies in a known compact set, then one can extend the parameter as $\Theta=( {\boldsymbol \psi}, {\boldsymbol \eta}, {\boldsymbol \nu},N).$ The following result shows the consistency of the ML estimator.

\begin{theorem} Under Assumption 3.2, the ML estimator is  consistent,
\[ \hat \Theta \stackrel{p}{\rightarrow} \Theta_0\] as $n\to \infty.$ 
\end{theorem}
\begin{proof} Here, we assume that the covariate ${\bf X}$ is a continuous random vector with pdf $g$, but one can get the same result with a discrete random vector. Also, we assume $N$ is known, but the same result follows similarly if it lies in a known compact interval.

Note that the fractional binomial distribution is identifiable, since, by (3.1), we have $E(B_N(p,H,c)) =Np,$
$P(B_N(p,H,c) =N)=p(p+c)^{N-1},$ and 
\[ var(B_N(p,H,c))=Np(1-p)+\sum_{\substack{i,j=1,2,\cdots,n\\ i\neq j}}pc|i-j|^{2H-2}.\] Therefore, $(p,H,c),$ and also $(p,H,c^{\circ})$, uniquely determine the first two moments and the probability of maximum observation, which indicates that the parameters uniquely determine the FB distribution. Combined with Assumption 3.2 {\it ii)}, it follows that  $\Theta$ uniquely determines the FB regression model, therefore, the identifiability of the FB regression model is established.

Under Assumption 3.1,  it follows from Proposition 2.2  in \cite{Lee} that   $P(B_N(p,H,c) =\ell)\neq 0$ for any $\ell=0,1,\cdots,N.$
Note that from (3.1-3.3) the log of the pmf of a fractional binomial distribution is continuous with respect to $(p,H, c^{\circ}),$ and from (3.4-3.6), $(p_i,H_i, c^{\circ})$ is continuous with respect to the parameter $\Theta$ and the covariates ${\bf x}_i.$ Therefore, $\ln P(y_i
| {\bf x}_i, \Theta)$ is continuous in $\Theta$ and  ${\bf x}_i$, and by Assumption 3.2 {\it i,iii},   
\[    \sup_{\Theta \in \mathcal{C}, {\bf x}_i \in \mathcal{X} } |\ln P(y_i| {\bf x}_i, \Theta)|:=D(y_i)  <\infty   , \] and
 $ E(D(y_i))=\int \sum_{y=0}^N P(y|{\bf x}, \Theta_0 )D(y) dg(\bf{x})<\infty. $
By Lemma 2.4 in \cite{Whi} \[  \sup_{\Theta\in \mathcal{C} }|\hat\ell(\Theta)-\ell(\Theta)|\stackrel{p}{\rightarrow} 0  \] where
\[  \hat \ell(\Theta| (y_i, {\bf x}_i ), i=1,2,\cdots, n)=\sum_{i=1}^n \ln P(y_i|{\bf x}_i, \Theta) \] and \[\ell(\Theta)=\int \sum_{y=0}^N  \ln P(y|{\bf x}, \Theta)
    P(y|{\bf x}, \Theta_0) dg({\bf x}),\]
    and by Theorem 2.5 in \cite{Whi}  \[\hat \Theta \stackrel{p}{\rightarrow} \Theta .\]
\end{proof}
\begin{corollary}
    The ML estimator of the fractional binomial distribution is  consistent if
    the parameters lie in a known compact set, $ (p,H,c^{\circ})\in \mathcal{D}$ where $\mathcal{D}$ is a compact set in  $( 0,1)^3.$
\end{corollary}
The consistency of the MLE is held under the set of assumptions, in particular, it is assumed that the parameters lie in a known compact set. 
In the next section, we check the performance and consistency of ML estimators  through simulations under Assumption 3.2.
However, in practice, the assumptions often fail as one may not know the possible range of parameters.  
Yet,  the consistency of estimators is not of concern in data analysis, as there is no guarantee that the data is generated from the assumed model. Rather, the goal is to find a model that is most suitable for data. Therefore, when we fit the FB regression model to data in Section 5, we do not limit the parameter space to a compact set to find MLEs.
\section{Simulation}
\subsection{Large sample properties of the MLE}
We simulated  fractional binomial regression models obtaining $y_i, i=1,2,\cdots, n=100,$ with a single covariate that follows uniform distribution, $x_{i}  \sim U(-2,2)$,  parameters $\Theta=(\psi_1 ,\psi_2, \eta_1, \eta_2, \nu_1, \nu_2  )=(-1, 1,2,1,0,-1 ), $ and $N=10.$ With the simulated data $\{ (y_i, x_i), i=1,2,\cdots,n\},$ we obtained the MLEs of the parameters with the constraint $\hat \Theta\in [-5,5]^6$ through numerical optimization. More specifically, we used the optimization method ``L-BFGS-B" in R, which is a variant of quasi-Newton method that allows constraints on parameters.  We repeated this procedure 20 times and computed the bias and standard error of the estimators. 
In the same way, the bias and standard error of the estimators were obtained with different values of $n,$ $\Theta$, and $N,$ and the results are shown in Table 1. 

Overall, the estimates are fairly close to the true parameters as the bias and standard error are less than .5 in most cases. 
It is observed that, for a given set of parameters $\Theta$ and $N,$ if we increase the sample size $n,$ the standard error of each estimator decreases, and the bias of each estimator also tends to decrease with few exceptions.  This shows that the larger the sample size, the more precise the estimator,  which is expected from the consistency of MLEs in Theorem 3.3.

We also simulated the FB model with  a vector covariate $ {\bf x}_i =(x_{1i},x_{2i})$ where $x_{ji}$ follows uniform distribution, $x_{ji}\sim U(-1,1),$ $ j=1,2,$ with the following two cases: (a) $x_{1i}$ and $x_{2i}$ are independent of each other ($\rho=0$) and (b) $x_{1i}$ and $x_{2i}$ are correlated ($\rho\approx 0.7$).
The parameter values were $\Theta=(\psi_1 ,\psi_2, \psi_3, \eta_1, \eta_2,\eta_3,  \nu_1, \nu_2 , \nu_3 )=(1, -2, 0.5,-0.5,2,1,2,1,3 ),$ and $  N=10.$ The bias and standard error of the MLEs are shown in Table 2.  In both cases of independent and correlated covariates, estimates become closer to the true parameters as the sample size $n$ increases from 100 to 400, since the bias and standard error of the estimators tend to decrease. 

We note the computational challenges posed by the FB model. Since the pmf of the fractional binomial distribution does not have a closed-form expression and should be computed iteratively for each set of parameters, simulating the FB model and obtaining the MLE through numerical optimization become  computationally expensive and take quite a long time if we increase the number of covariates and the bound on the maximum response. Through simulations and more experiments with data analysis in the next section, it is found that the practicality of the FB model is limited to small datasets with only a few covariates (approximately fewer than 5) and a small or moderate bound on the maximum response (approximately less than 20). The sample size seems to have a lesser effect on the computation time, and its effect is confounded with other factors, the number of covariates and the bound on the maximum response. For example, repeated observations increase the sample size but may not increase the computational burden as much.

\begin{table}[h]
    \centering
    \begin{tabular}{|c|c||c|c||c|c|}
    \hline
 \multicolumn{2}{|c||}{ }&   \multicolumn{2}{c||}{ N=10}&  \multicolumn{2}{c|}{ N=20} \\ 
 \multicolumn{2}{|c||}{ }     & n=100 & n=400& n=100 & n=400 \\ \cline{3-6} 
      \multicolumn{2}{|c||}{ }    &Bias (SE)&  Bias (SE) & Bias (SE)& Bias (SE)  \\ \hline
     $ \psi_1 =-1$ & $ \hat\psi_1 $ &   0.01 (0.23)& 0.00 (0.07) &  0.06 (0.13)&   0.00 (0.06)  \\
    $\psi_2 =1$   & $ \hat \psi_2 $ & 0.03 (0.16)& -0.00 (0.06) &    0.00 (0.10)&  0.02 (0.06)  \\ \cline{2-6}
     $\eta_1 =2 $ & $\hat \eta_1$ & 0.26 (1.06)& -0.03 (0.33)&   -0.04 (0.42)&  0.01 (0.24)   \\
      $\eta_2 =1$   &$\hat \eta_2$ &0.21(0.70) &  -0.03 (0.21) &  -0.07 (0.28)& 0.06 (0.13)  \\ \cline{2-6}
    $\nu_1 =0$    &$\hat \nu_1$& 0.23 (1.07)&  0.15 (0.26)&   0.12 (0.57)& -0.03 (0.30) \\
     $\nu_2 =-1$    &$\hat \nu_2$&  -0.31(0.99)&  -0.02 (0.27)&   -0.36 (0.65)& -0.06 (0.31) \\ \hline \hline
        $\psi_1 =-0.5$& $\hat\psi_1$ & -0.03 (0.24)  & -0.03 (0.09) &-0.08 (0.22) &  -0.01 (0.09)   \\
        $\psi_2 =2$ &  $\hat \psi_2$ &0.06 (0.27) & 0.02 (0.07) & 0.14 (0.26)   &  0.05 (0.11) \\ \cline{2-6}
        $\eta_1 =1$ &  $\hat \eta_1$ &-0.05 (0.63) & 0.02 (0.24) & 0.03 (0.50) &    -0.07 (0.18) \\
        $\eta_2 =-2$ & $\hat \eta_2$  &-0.72 (1.35)& -0.13 (0.37)   &  -0.44 (0.96)&  0.01 (0.35) \\ \cline{2-6}
        $\nu_1 =2$ &  $\hat \nu_1$  &0.66 (1.47) & 0.31 (0.75) & 0.55 (1.56)  & 0.29 (0.68) \\
        $\nu_2 =1$ & $\hat \nu_2$ &-0.28(1.35) & 0.30 (0.64) & -0.14 (1.32)  &   0.26 (0.51)  \\ \hline
    \end{tabular}
    \caption{Bias and standard error of MLEs with a single covariate}
    \label{tab:my_label}
\end{table}

\begin{table}[h]
    \centering
    \begin{tabular}{|c|c||c|c||c|c|}
    \hline 
    \multicolumn{2}{|c||}{ }& \multicolumn{4}{c|}{ N=10} 
    \\ \cline{3-6} 
 \multicolumn{2}{|c||}{ }&   \multicolumn{2}{c||}{ $\rho=0$ }&  \multicolumn{2}{c|}{ $\rho\approx .7$} \\ 
 \multicolumn{2}{|c||}{ }     & n=100 & n=400& n=100 & n=400 \\ \cline{3-6} 
\multicolumn{2}{|c||}{ }    &Bias (SE)&  Bias (SE) & Bias (SE)& Bias (SE)  \\ \hline
     $ \psi_1 =1$ & $ \hat\psi_1 $ &   0.05 (0.18)& 0.00 (0.05) &  0.06 (0.14)&   0.01 (0.05)  \\
    $\psi_2 =-2$   & $ \hat \psi_2 $ & -0.02 (0.22)& 0.02 (0.16) &    0.12 (0.43)&  0.02 (0.16)  \\ 
    $\psi_3 =0.5$   & $ \hat \psi_3 $ & 0.03 (0.18)& -0.01 (0.10) &    -0.18 (0.61)&  -0.06 (0.20)  \\ \cline{2-6}
     $\eta_1 =-0.5 $ & $\hat \eta_1$ & -0.27 (1.47)& 0.17 (0.44)&   -0.27 (0.90)&  -0.11 (0.45)   \\
      $\eta_2 =2$   &$\hat \eta_2$ &0.52 (1.76) &  -0.29 (0.45) &  0.57 (1.57)& 0.33 (0.67)  \\ 
      $\eta_3 =1$   &$\hat \eta_3$ &0.03 (1.23) &  0.00 (0.54) &  0.63 (1.92)& -0.13 (0.88)  \\ \cline{2-6}
    $\nu_1 =2$    &$\hat \nu_1$& 0.17 (1.19)&  0.24 (0.53)&   -0.34 (0.88)& -0.10 (0.48) \\
     $\nu_2 =1$    &$\hat \nu_2$&  0.38 (2.00)&  0.04 (1.26)&   -0.48 (2.03)& -0.08 (1.18) \\ 
     $\nu_3 =3$    &$\hat \nu_3$&  0.89 (1.45)&  0.55 (1.00)&   -0.10 (2.44)& -0.12 (1.23) \\ \hline 
    \end{tabular}
    \caption{Bias and standard error of MLEs with a vector covariate}
    \label{tab:my_label}
\end{table}

\subsection{Comparing model performance}

To compare the performance of ZI models and the FB model, data from each of the models were simulated. 
The FB regression model was simulated, $\{(y_i, x_i), i=1,2,\cdots,n\}$,  with sample size $n=200$, $N=10$, and a single covariate from the uniform distribution (0,1). The simulation was repeated 20 times with  parameters randomly generated from a uniform distribution such that  $p_i\in(0.12, 0.95), $ and 
$ h_i, c_i^{\circ}\in (0.50, 0.95),$ for all $i=1,\cdots, n.$ The range of the parameters was set to ensure the presence of zero-inflation in the simulated data with various centers and distribution dispersions. In the same way,  $\{(y_i, x_i), i=1,2,\cdots,200\}$ were generated from ZIB such that $p_i\in (0.12,0.95), \pi_i\in (0.03,0.62),  $ and from ZIP, ZINB, ZINB-2 with parameters in the range of   $\lambda_i\in(0.02, 4.48),$ $ 
 \theta_i\in(0.08, 12.18 ), $ $ \mu_i\in (0.82, 3.00),$ and $\pi_i\in(0.03, 0.62) .$ The parameter range was chosen so that the response variable is not too large, approximately bounded by 15. 
 \begin{table}[htp]
    \centering
    \begin{tabular}{|c|c|c|c|}
   \hline \\[-1em] True model: FB   &      $\overline{\Delta}$  AIC&  \% $\Delta$  AIC $>4$  & \% $\Delta$  AIC $< -4$       \\ \hline
AIC(ZIB)-AIC(FB)& 86.43& 100\% & 0\%\\
 AIC(ZIP)-AIC(FB) &  142.02&  95\% & 0\%  \\
AIC(ZINB)-AIC(FB)  & 143.04 &    95 \% & 0\% \\
AIC(ZINB2)-AIC(FB)&   152.42&   100 \% & 0\% \\ 
   \hline \\[-1em]
    True model: ZIB & $\overline{\Delta}$  AIC&  \% $\Delta$  AIC $>4$  & \% $\Delta$  AIC $< -4$      \\  \hline
    AIC(ZIB)-AIC(FB) &  -1.33&  5\% & 15 \%   
    \\ 
    AIC(ZIP)-AIC(FB) &  46.11&  75\% & 0 \%   
    \\ 
AIC(ZINB)-AIC(FB) &  48.11&  85\% & 0 \%   
    \\ 
    AIC(ZINB2)-AIC(FB) &  54.64&  90\% & 0 \%   
    \\ \hline \\[-1em]
    True model: ZIP  &      $\overline{\Delta}$  AIC&  \% $\Delta$  AIC $>4$  & \% $\Delta$  AIC $< -4$      \\ \hline
     AIC(ZIB)-AIC(FB) & 10.71 & 75 \%  & 0\%  \\
        AIC(ZIP)-AIC(FB) &  -1.27&  15\% & 25 \%   \\
AIC(ZINB)-AIC(FB)  & -0.19 &    10 \% & 10 \%\\
AIC(ZINB2)-AIC(FB)&   2.24&   20 \% & 0\%\\ \hline \\[-1em]
 True model: ZINB  &      $\overline{\Delta}$  AIC&  \% $\Delta$  AIC $>4$  & \% $\Delta$  AIC $< -4$      \\ \hline
        AIC(ZIB)-AIC(FB) & 93.72 &  100\%  & 0\%  \\
        AIC(ZIP)-AIC(FB) &  37.08&  80\% & 5\%   \\
AIC(ZINB)-AIC(FB)  & -9.68 &    0 \% & 45 \%\\
AIC(ZINB2)-AIC(FB)&   -8.66&   0 \% & 45\%\\ \hline \\[-1em]
         True model: ZINB-2  &      $\overline{\Delta}$  AIC&  \% $\Delta$  AIC $>4$  & \% $\Delta$  AIC $< -4$      \\ \hline
 AIC(ZIB)-AIC(FB) & 89.18 &  95\%  & 0\%  \\
        AIC(ZIP)-AIC(FB) &  25.18&  80\% & 5\%   \\
AIC(ZINB)-AIC(FB)  & -7.13 &    0 \% & 60 \%\\
AIC(ZINB2)-AIC(FB)&   -5.58&   0 \% & 45\%\\ \hline  
    \end{tabular}
    \caption{Comparing model performance using AIC (${\Delta}$ AIC denotes the AIC difference between two models, and    $\overline{\Delta}$ AIC denotes the mean AIC difference between two models.) }
    \label{tab:my_label}
\end{table}

The ZI regression models and the FB model were fitted to the simulated data of each model, and the model fit was compared using AIC in Table 3. 
The results show that when the true generating process is the FB model, FB fits data better than other models, as the AIC of the FB model is lower on average by at least 86.43 points than the AICs of other ZI models, and 95\% of the time or more the AIC of the FB is lower by more than 4 points than the AICs of the other models.  
When a ZI model is the true generating process, FB does not perform as well as the true model, but still performs better than some of the other ZI models.  
For example, when ZINB was the actual generating process,  ZINB and ZINB-2 better fit the data, as FB has a higher mean AIC than ZINB and ZINB-2, with 45\% of the time having the AIC greater by more than 4 points and no case of the AIC lower by more than 4 points than ZINB and ZINB-2. However, the FB model still performs better than ZIB and ZIP as the AIC of the FB is lower by more than 4 points 80\% of the time or more than those of ZIB and ZIP. 
\par
To understand when FB performs better than other models, a graphical examination is provided in Figure 2, in which the data distribution and the fitted probability distributions are overlaid for each simulated count dataset of the fractional binomial distribution. Note that it is when the distribution is left-skewed that the FB distribution fits the data much better than other ZI distributions. 
\begin{figure}[htp]
    \centering
\includegraphics[width=0.45\linewidth]{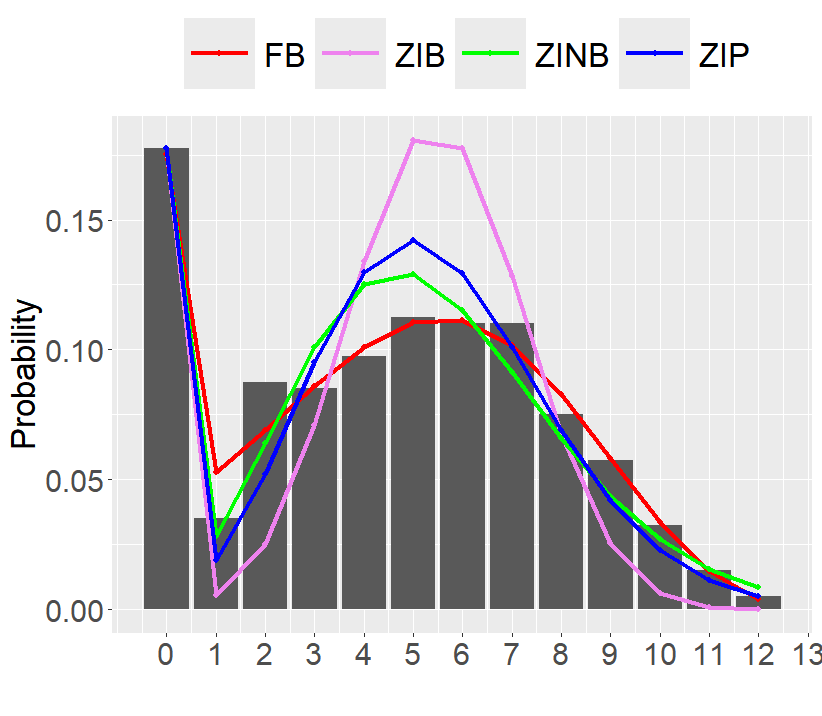}
\hspace{7pt}
\includegraphics[width=0.45\linewidth] 
{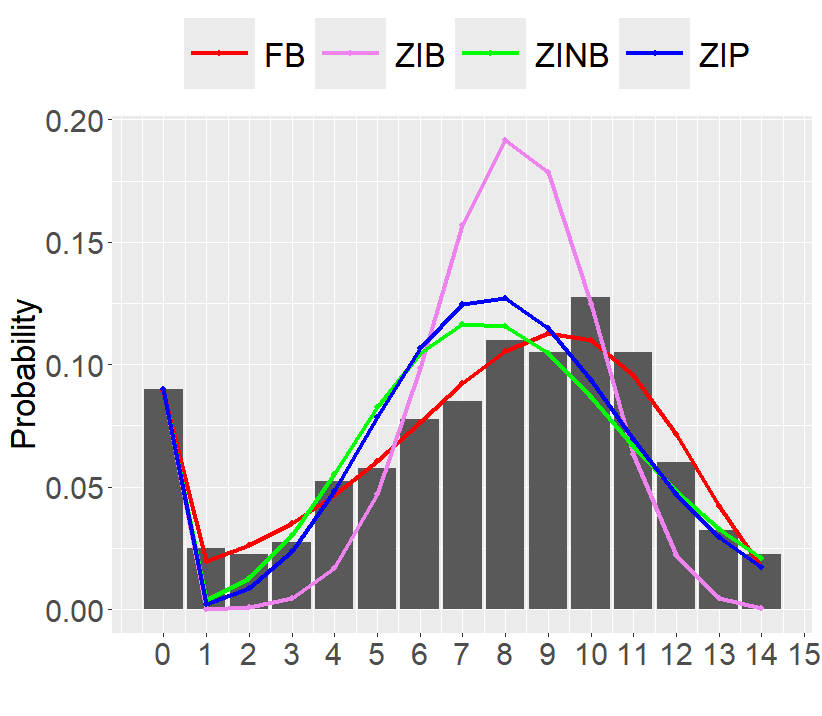}
\\
\includegraphics[width=0.45\linewidth]{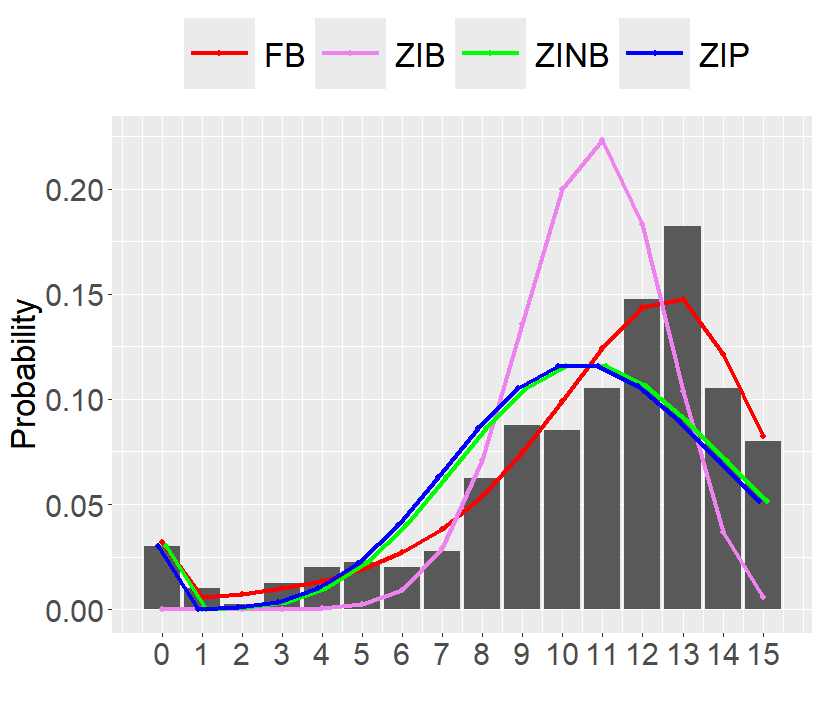}
\hspace{7pt}
\includegraphics[width=0.45\linewidth]{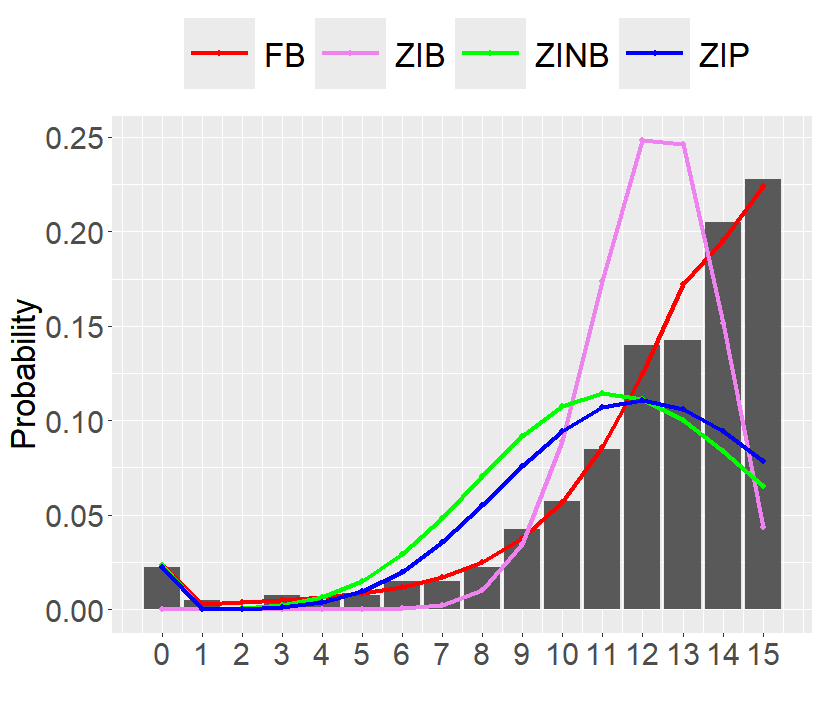}
    \caption{Fitted probability distributions for simulated count data generated by fractional binomial distribution with parameters $(p,H,c)\in (0.3, 0.7, 0.2 )$  for top left, $(p,H,c)\in (0.5, 0.7, 0.2 )$ for top right, 
    $(p,H,c)\in (0.7, 0.7, 0.15 )$ for bottom left, and $(p,H,c)\in (0.8, 0.6, 0.1)$ for bottom right.}
    \label{fig:enter-label}
\end{figure}


 \section{Data analysis}
 \subsection{Data from horticulture}
 In \cite{Rid}, zero-inflated models were fitted to a dataset that is from an experiment published in 1993 titled ``Micropropagation of columnar apple trees." In their research, they micropropagated 270 shoots of the columnar apple cultivar Trajan. These shoots were treated with various experimental conditions, and the number of roots for each apple shoot was recorded. More specifically, the apple shoots were treated under differing concentrations of cytokinin 6-benzylaminopurine (BAP) and different photo periods; some were treated with an 8-hour photo period, whilst the rest were treated with a 16-hour photo period.  We fitted the FB regression model and other zero-inflated models (ZIB, ZIP, ZINB, ZINB-2) to the dataset with the response variable, the number of roots for each apple shoot, and a numerical covariate, BAP, and a categorical covariate, photoperiod (Pho).


\begin{table}[htp]\begin{tabular}{|c|c|ccc|}\hline   \\[-1em]   log/logit&Pho(categorical)  &ZIP  & ZINB   &ZINB-2   \\ link& BAP(numerical) & Coef (p-value)& Coef (p-value)& Coef (p-value) \\
\hline &  (Intercept)& 1.97 (0.00)& 1.98 (0.00) & 1.96 (0.00)  \\ $\mu$      &Pho(16) & -0.28 (0.00)& -0.28 (0.00) &-0.39 (0.00)\\(log)     & BAP  & 0.00 (0.92) & 0.00 (0.92) & 0.00 (0.66) \\\hline&  (Intercept) & -4.31 (0.00) & -4.52 (0.00) & -4.43 (0.00)\\$\pi$       &Pho(16) & 4.18 (0.00)  & 4.40 (0.00) & 4.08 (0.00) \\ (logit)     & BAP & 0.00 (0.92) & 0.00 (0.99) & 0.02 (0.60) \\\hline  &(Intercept) & &2.51 (0.00) & 1.56 (0.02)\\$\theta$       &Pho(16) & &  & -1.52 (0.03)\\ (log)  &  BAP  & & &0.20 (0.06)  \\\hline\end{tabular}
\caption{Zero-inflated models with apple data}\end{table}

{\begin{table}[htp]\begin{tabular}{|c|c|c|c|}\hline \\[-1em]   & Pho(categorical) & ZIB &FB  
\\ logit link& BAP(numerical)& Coef (p-value)& Coef (p-value)  \\
\hline  &  (Intercept)& 
  -0.32 (0.00)& -0.37 (0.00) \\ $p$     &Pho(16) & -0.44 (0.00) & -1.29 (0.00)  \\ & BAP & 0.00 (0.88) & 0.00 (0.64)
  \\ \hline 
  &(Intercept)& 
  -4.27 (0.00)&  \\ $\pi$     &Pho(16) & 4.15 (0.00) &  \\ & BAP & 0.00 (0.90) & 
  
  \\ \hline&  (Intercept)&& -1.47 (0.10)\\ $H$       &Pho(16)& & 2.53 (0.01) \\& BAP& &0.04 (0.42)   \\ \hline  &(Intercept)& & 2.46 (0.05) \\ $c^{\circ}$       &Pho(16)&& -0.67 (0.56) \\ &  BAP& &-0.14 (0.04) \\ \hline\end{tabular} \caption{ZIB and Fractional binomial regression models with apple data}\label{tab:my_label}\end{table}

Tables 4 and 5 show the estimates of the coefficients and their p-values for each model.  When comparing the results of the three models, ZIP, ZINB, and ZINB-2, most of the additional coefficients in the dispersion parameter $\theta$ in ZINB and ZINB-2 are significant at the 5\% significance level.  It is also observed that the covariate BAP is not statistically significant  in all zero-inflated regression models (ZIB, ZIP, ZINB, ZINB-2) except in the FB model, where BAP has a p-value of 0.04 
in $c^{\circ}$, which leads us to fit all models again without BAP. 

\begin{table}[htp]
     \centering
     \begin{tabular}{|c|c|c|c|c|c|} \hline
     covariates: Pho, BAP &ZIB & ZIP& ZINB& ZINB-2&FB \\ \hline 
   Log-likelihood  & -677.18 & -630.64& -621.95& -615.00&-611.48  \\  \hline 
    AIC& 1366.36&  1273.28 & 1257.90& 1248.00 & 1240.96 \\  
\hline      \multicolumn{5}{c}{ } \\   \hline
     covariate: Pho & ZIB&ZIP& ZINB& ZINB-2&FB \\ \hline 
     Log-likelihood & -677.20 & -630.60& -622.00 & -618.81&  -614.53 \\  \hline 
     AIC& 1362.40 &1269.30 & 1253.92& 1249.63 &  1241.06 \\ 
\hline
     
     \end{tabular}
     \caption{Model comparison}
     \label{tab:my_label}
 \end{table}

 Table 6 shows the AICs of the fitted models with and without BAP.  For each set of covariates, the FB model has the lowest AIC among all the models, indicating that the FB model fits the data best. For each model, there is not much difference in AIC after the covariate BAP is removed. For FB and ZINB-2, the difference in AIC is less than 2, and for all other models it is less than 4.  Therefore, we include Pho only for the covariate and proceed to model checking.

 \begin{figure}[htp]
     \centering
\includegraphics[width=0.32\linewidth]{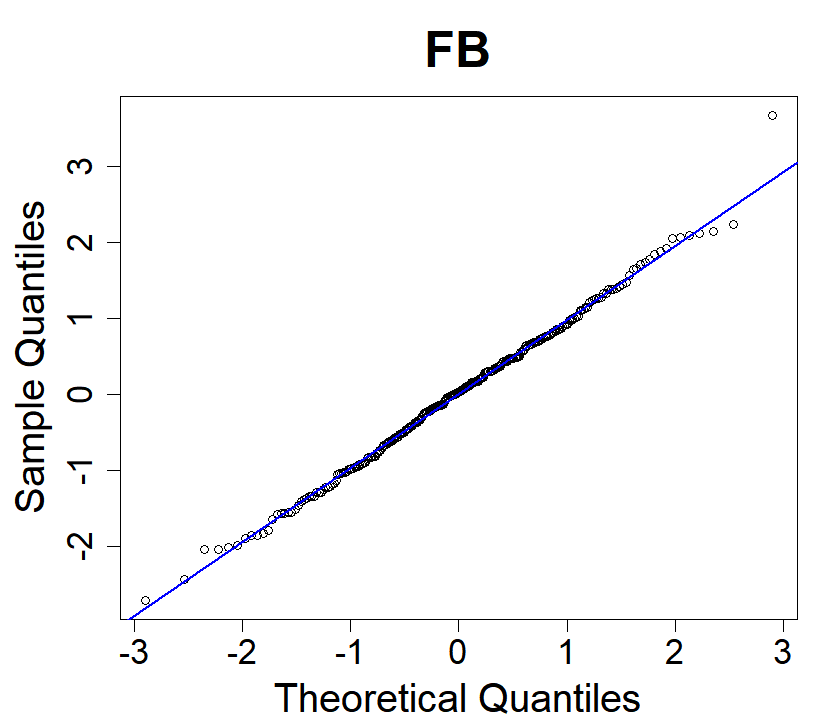}\includegraphics[width=0.32\linewidth]{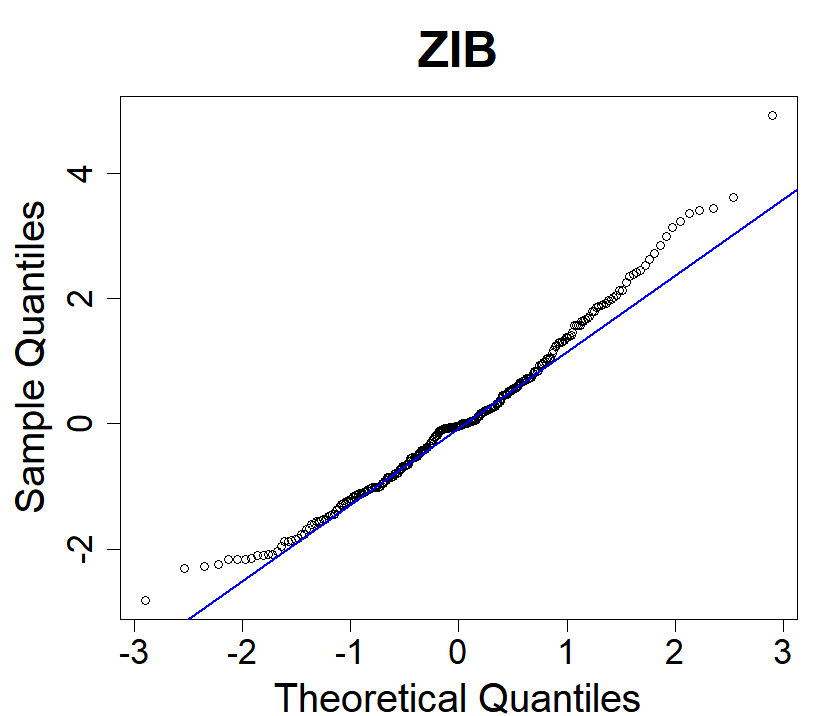}\includegraphics[width=0.32\linewidth]{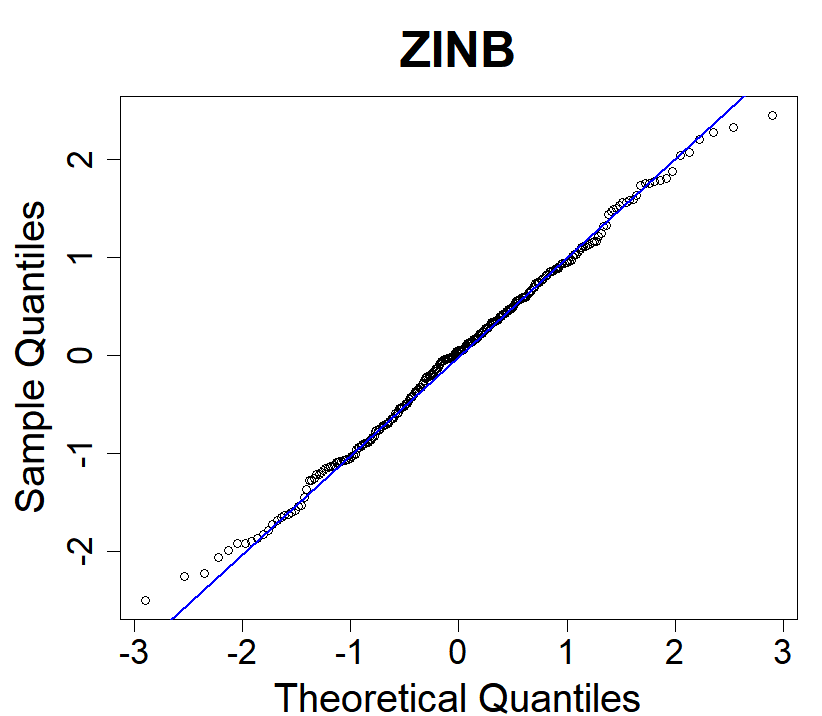}

\includegraphics[width=0.32\linewidth]{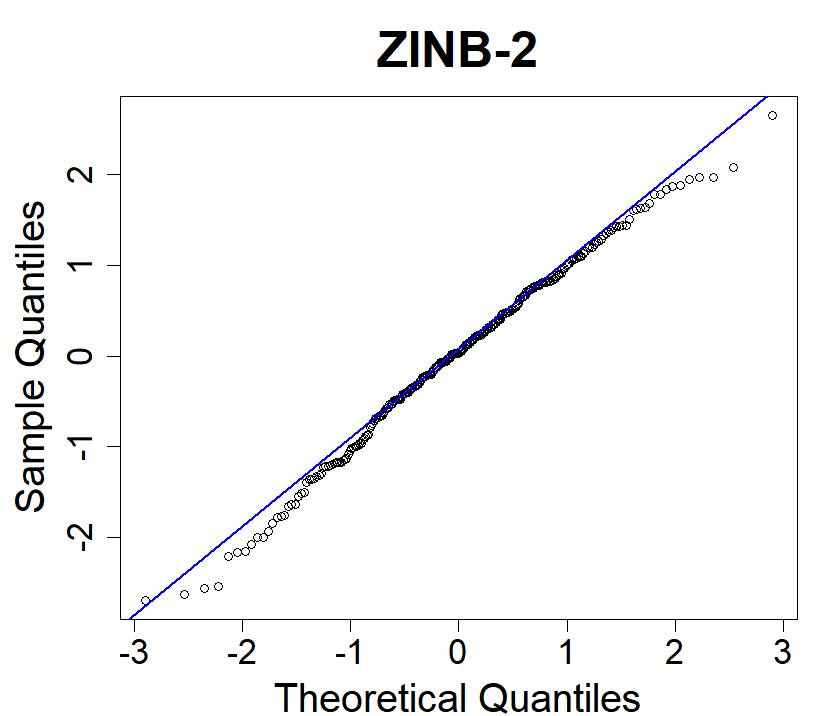}\includegraphics[width=0.32\linewidth]{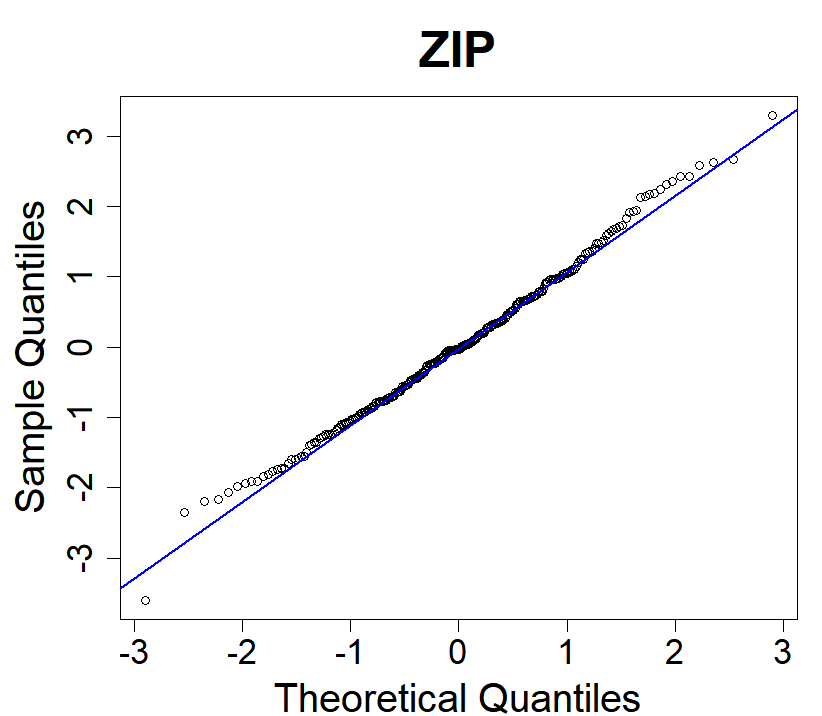}
     \caption{Q-Q plot for Randomized Quantile Residuals}
     \label{fig:enter-label}
 \end{figure}

\begin{table}[htp]
     \centering
    \begin{tabular}{|c|c|c|c|c|c|}
       \hline  Model &  ZIB& ZIP &  ZINB& ZINB-2 &FB \\ 
         \hline 
        p-value &  0.00 & 0.52 & 0.52 &  0.40& 0.75\\
        \hline 
    \end{tabular}
    \caption{ P values for the SW normality test of RQRs }
    \label{tab:my_label}
\end{table}
 
We used randomized quantile residual (RQR)  for model diagnostics, which was first proposed in \cite{Dun} for discrete data, and was used for zero-inflated count models in \cite{bai, feng, Sell1}. If a model is correctly specified, RQR should be approximately normally distributed, and model misspecification can be detected by departure from the normality of RQR. We used the Q-Q plot (Figure 3) and the Shapiro-Wilk (SW) normality test (Table 7) to examine the normality of RQR. For each model, the average p-value of the SW test was obtained from 100 replicated RQRs, since RQR involves randomness \cite{bai}.
All models except ZIB pass the normality test at the 5\%  significance level, with the FB model having the largest p-value.  The fitted probability distribution and data distribution for each value of the covariate are shown in Figure 4. It is observed that the FB model fits better, while other models underestimate the probability, at response value 1, in both cases.


\begin{figure}[htp]
    \centering
\includegraphics[width=0.45\linewidth]{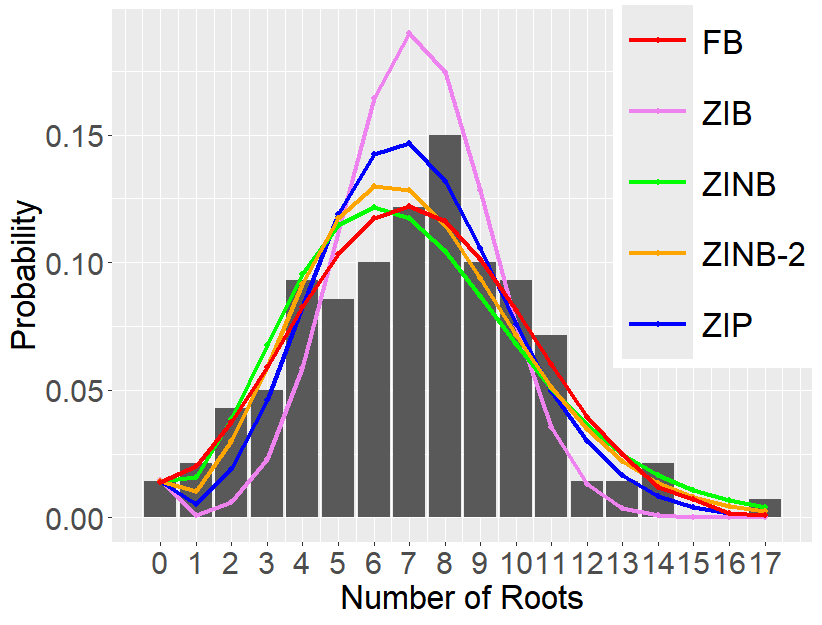}
\hspace{8pt}
\includegraphics[width=0.45\linewidth]{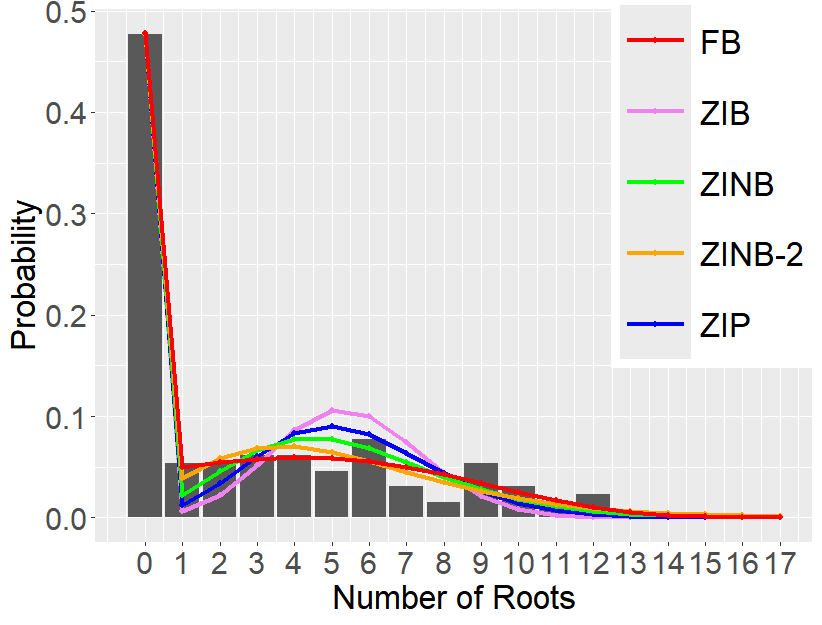}
    \caption{Data distribution and fitted distribution for photoperiod=0 (left), and photoperiod=1 (right).}
    \label{fig:enter-label}
\end{figure}

 \subsection{Data from public health}
Mortality dataset from the Human Mortality Database contains death count during 1933-2023 in the U.S. (\url{ https://www.mortality.org/Country/Country?cntr=USA}), which was originally from the U.S. Census Bureau and the National Center for Health Statistics.
 In particular, we used a dataset that recorded estimates of the number of deaths per 10-year time interval and 5-year age group.

\begin{figure}[h]
    \centering
\includegraphics[width=.9\linewidth]{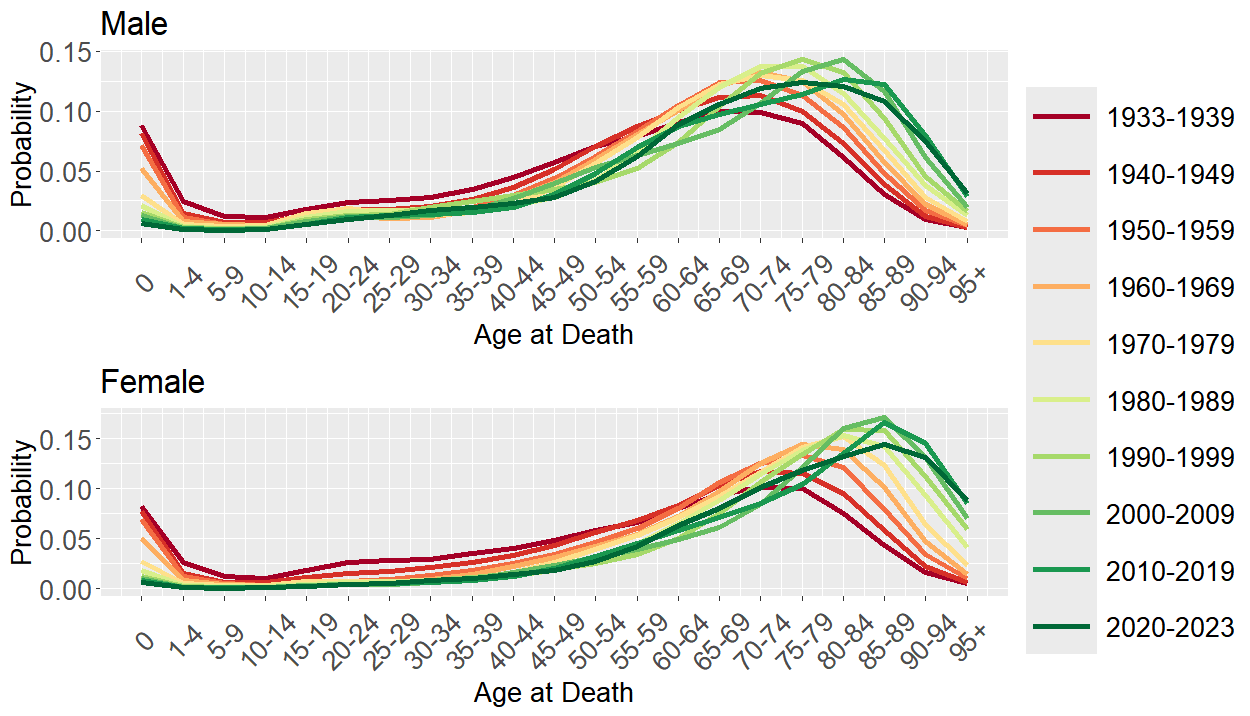}
    \caption{Data distribution of age at death}
    \label{fig:enter-label}
\end{figure}

The dataset is visualized in Figure 5. It shows that for both men and women, the distributions of age at death are zero-inflated, especially in the years 1933-1979, and left-skewed. The ZI models and the FB model were fitted to the dataset where the response variable is age at death with 5-year age group, i.e., $y=0$ for age at death 0,  $y=1$ for age at death between $1-4$, $y=2$ for  $5-10$, $\cdots,$ $y=19$ for $90-94$, $y=20$ for 95 or above. Sex and year of death (Year) were used as covariates with Year=1 for year of death between $1933-1939$, Year=2 for $1940-1949$, $\cdots$, Year=9 for $2010-2019,$ and Year=10 for $2020-2023.$ Since the size of the data is large as it is from censuses over many decades, we used sample data obtained by sampling one in every 10,000 observations in each of 10-year time interval, 5-year age group, and gender, e.g., during the years $1933-1939,$  the estimate of the death counts for females at age 0 was 362,056.18, in which  36 death counts was sampled.

\begin{table}[htp]\begin{tabular}{|c|c|ccc|}\hline \\[-1em]  log/logit&    &ZIP & ZINB &ZINB-2 \\ link& & Coef (p-value)& Coef (p-value)& Coef (p-value)    \\\hline &  (Intercept)& 2.50 (0.00)& 2.49 (0.00) & 2.51 (0.00)  \\ $\mu$      &Female & 0.08 (0.00)& 0.06 (0.00) &0.08 (0.00)\\(log)     & Year  & 0.02 (0.00) & 0.02 (0.00) & 0.02 (0.00) \\\hline&  (Intercept) & -1.77 (0.00) & -2.22 (0.00) & -1.77 (0.00)\\$\pi$       &Female & -0.09 (0.27)  & -0.24 (0.01) & -0.09 (0.27) \\ (logit)     & Year & -0.33 (0.00) & -0.24 (0.00) & -0.33 (0.00) \\\hline  &(Intercept) & &1.78 (0.00) & 0.99 (0.00)\\$\theta$       &Female & &  & -0.32 (0.11)\\ (log)  &  Year  & & &1.69 (0.00)  \\\hline\end{tabular}\caption{ZI models for mortality data}\end{table} }
{\begin{table}[htp]\begin{tabular}{|c|c|c|c|}\hline \\[-1em]    &  & ZIB &FB
\\ logit link& & Coef (p-value)& Coef (p-value)
\\\hline  &  (Intercept)& 
  0.35 (0.00)& 0.13 (0.00) \\ $p$     &Female & 0.30 (0.00) & 0.30 (0.00)  \\ & Year & 0.08 (0.00) & 0.10 (0.00)
  \\ \hline 
  &(Intercept)& 
  -1.77 (0.00)&  \\ $\pi$     &Female & -0.09 (0.27) &  \\ & Year & -0.33 (0.00) & 
  
  \\ \hline&  (Intercept)&& 1.65 (0.00)\\ $H$       &Female& & -0.00 (0.95) \\& Year& &-0.15 (0.00)   \\ \hline  &(Intercept)& & -0.15 (0.05) \\ $c^{\circ}$       &Female&& 0.10 (0.13) \\ &  Year& &0.09 (0.00) \\ \hline\end{tabular} \caption{ZIB and the FB models for mortality data}\label{tab:my_label}\end{table}

\begin{table}[htp]
     \centering
     \begin{adjustbox}{width=.8\columnwidth,center}
     \begin{tabular}{|c|c|c|c|c|c|} \hline
      &ZIB & ZIP& ZINB& ZINB-2&FB \\ \hline 
        Log-likelihood & -57957.04& -51707.28& -56652.07& -51607.46&-47943.51  \\ \hline
          AIC& 115926.1&103426.6 &113318.1 &103232.9 & 95905.02 \\ \hline 
     \end{tabular}
     \end{adjustbox}
     \caption{Model comparison}
     \label{tab:my_label}
 \end{table}


 The estimated coefficients and their p-values in each model are shown in Tables 8-9, and the AIC in Table 10.  FB  has the lowest AIC, followed by ZINB-2, with a difference in AIC between the two approximately as large as 7300. The fitted probability distributions overlaid with the data distribution for the selected years are shown in Figure 6. The results show that the FB model fits the data better compared to other zero-inflated models. The goodness of fit of the models was also checked by RQR. The Q-Q plot of RQRs in Figure 7 reveals a lack of fit of all models except the FB model, which is not surprising, as we have already seen that the FB distribution fits left-skewed data better than other ZI distributions at the end of Section 4.

 \begin{figure}[htp]
     \centering
\includegraphics[width=0.45\linewidth]{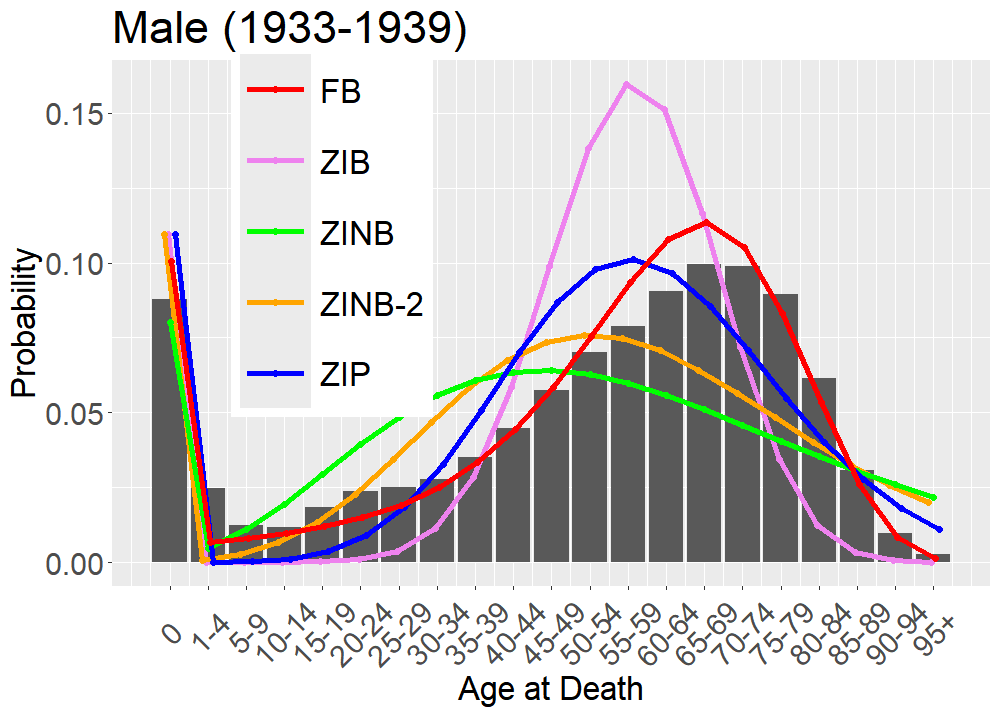}\hspace{8pt}\includegraphics[width=0.45\linewidth]{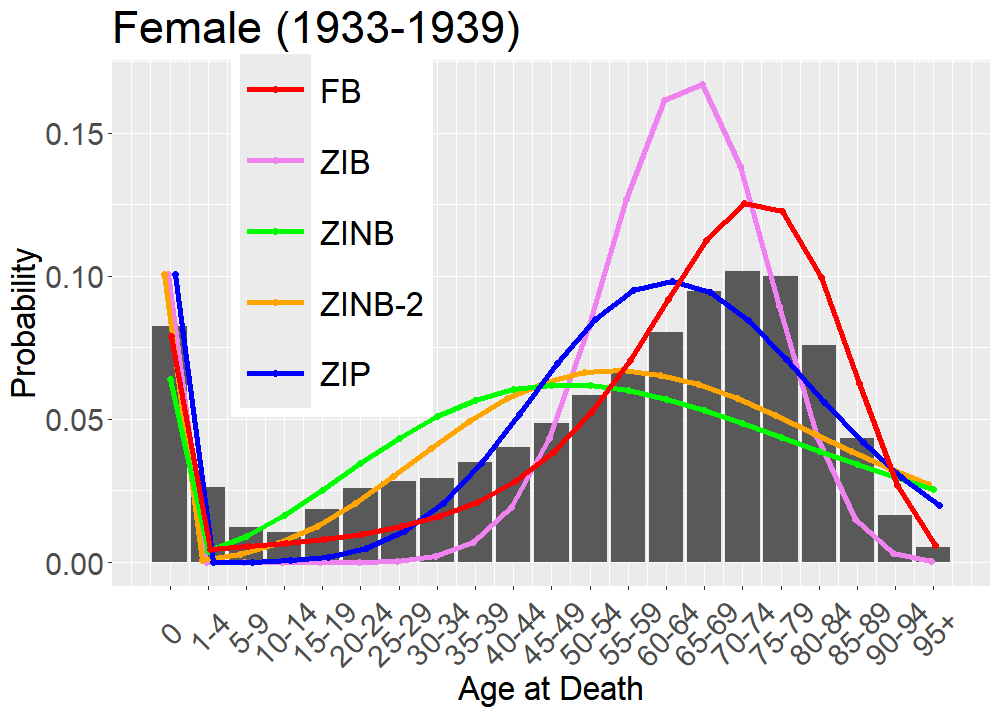}

\includegraphics[width=0.45\linewidth]{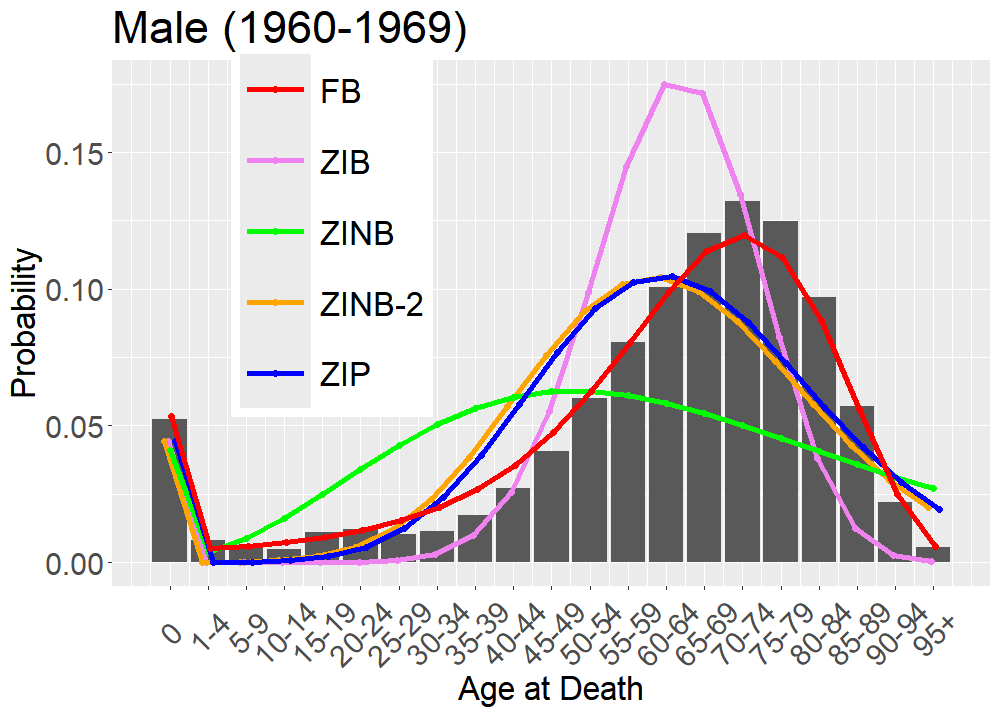}\hspace{8pt}\includegraphics[width=0.45\linewidth]{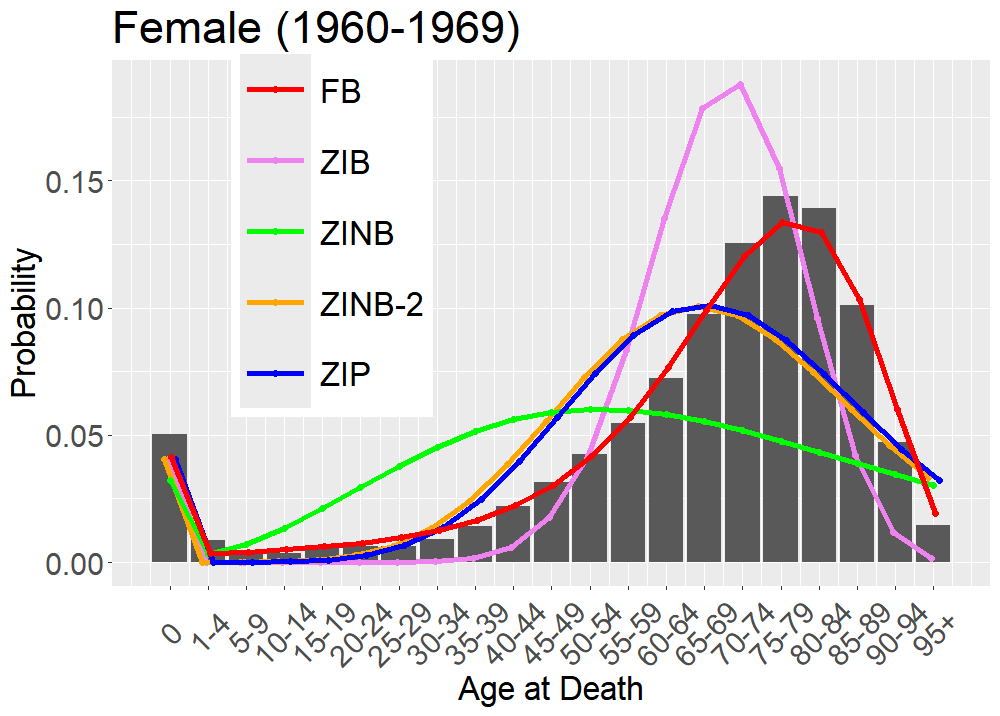}
     
     \includegraphics[width=0.45\linewidth]{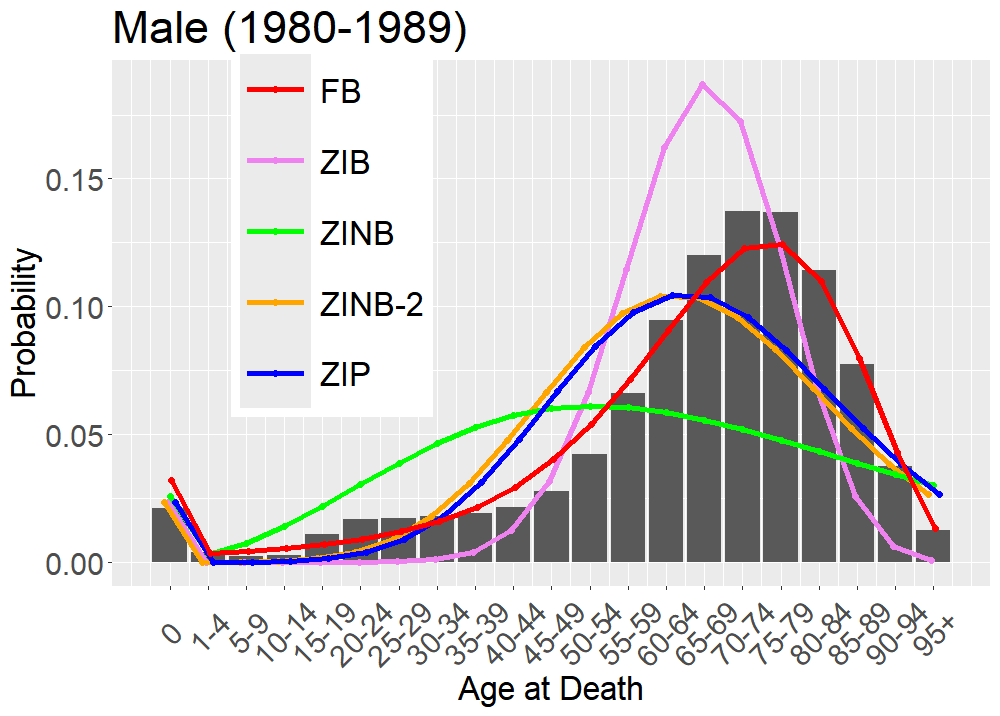}\hspace{8pt}\includegraphics[width=0.45\linewidth]{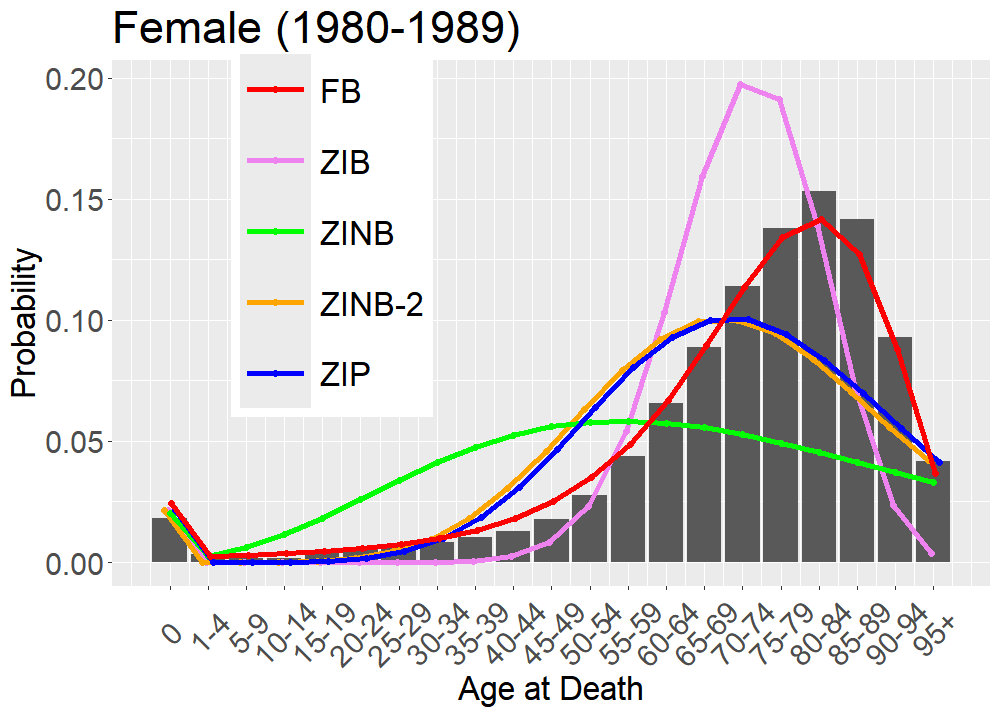}
     
     \includegraphics[width=0.45\linewidth]{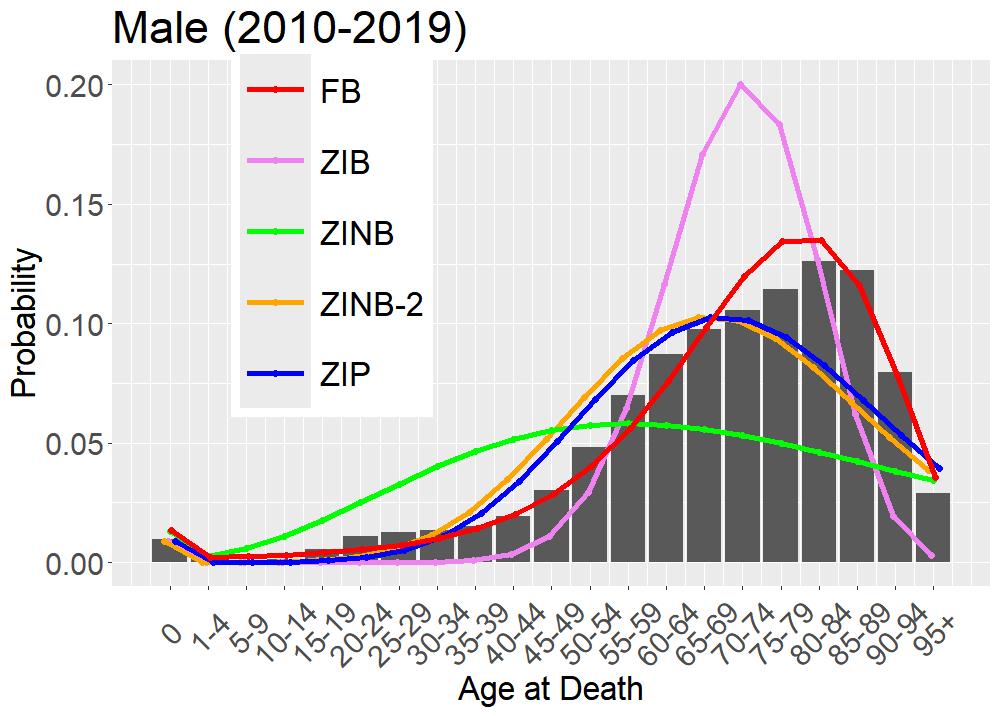}\hspace{8pt}\includegraphics[width=0.45\linewidth]{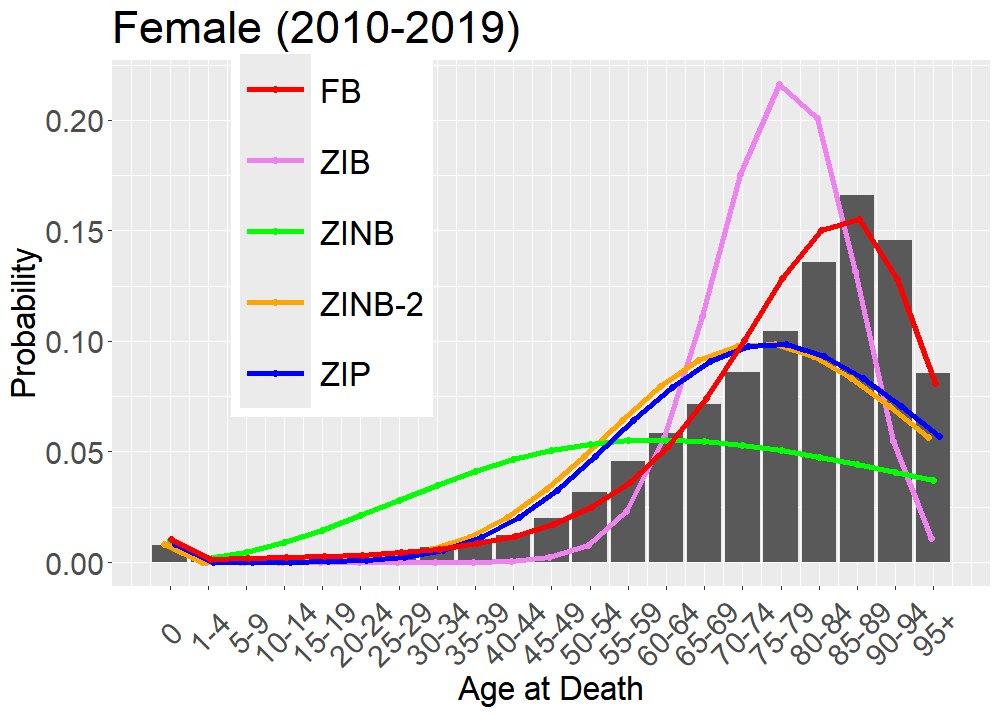}
     
     \caption{Data distribution overlaid with fitted probability distributions for selected years}
     \label{fig:enter-label}
 \end{figure}

\begin{figure}[htp]
    \centering
\includegraphics[width=0.32\linewidth]{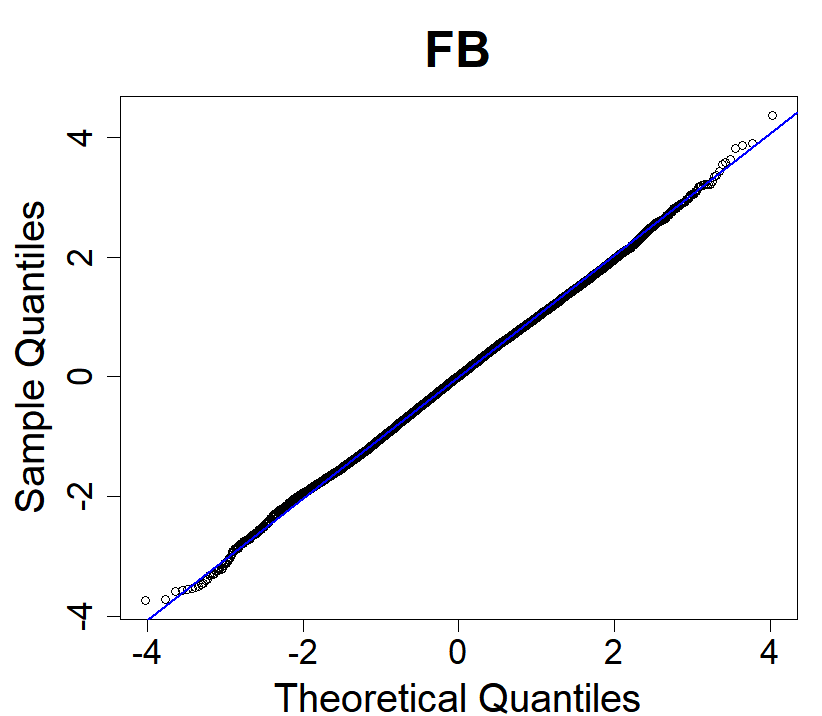}
\includegraphics[width=0.32\linewidth]{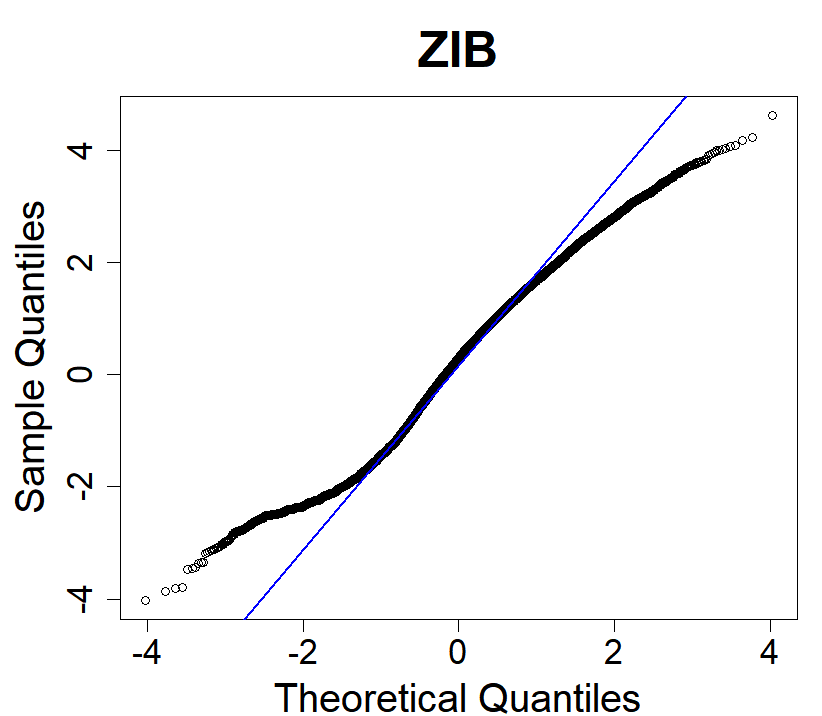}
\includegraphics[width=0.32\linewidth]{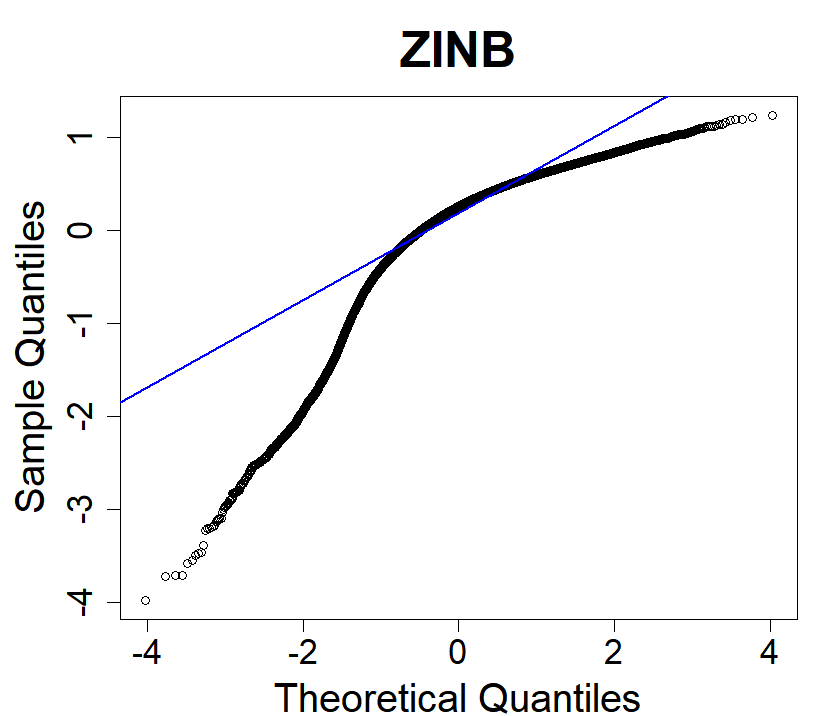}

\includegraphics[width=0.32\linewidth]{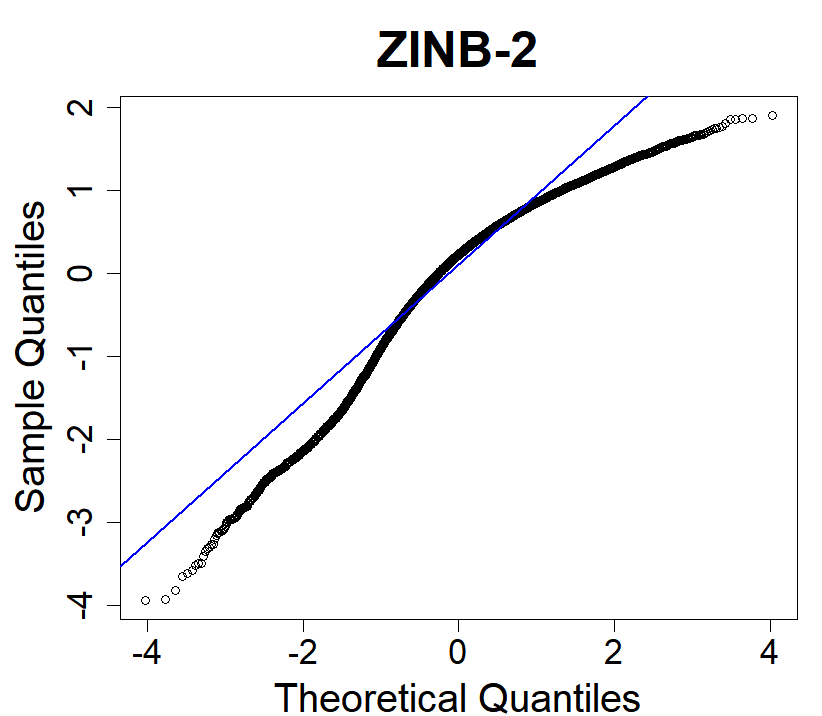}
\includegraphics[width=0.32\linewidth]{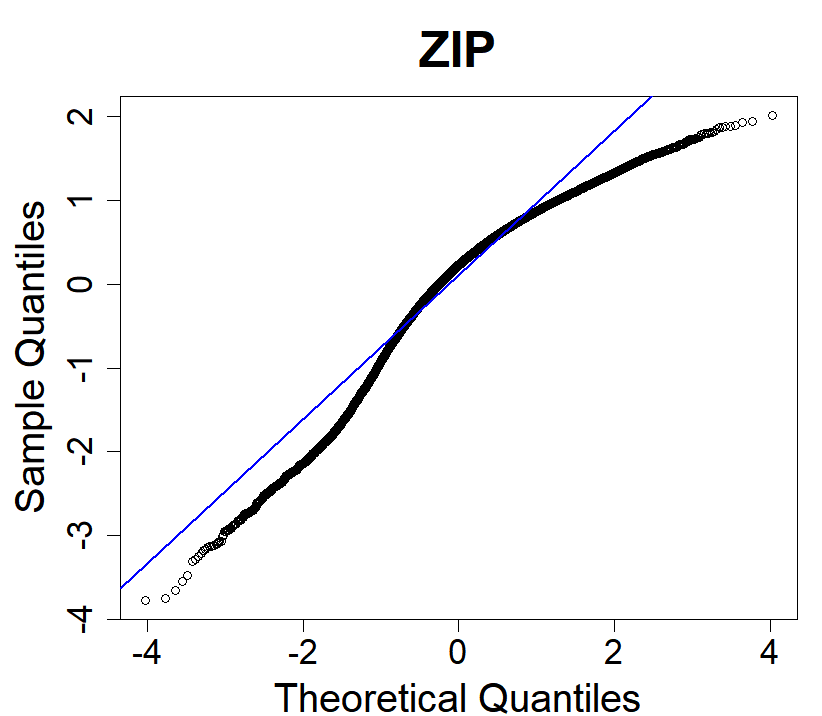}
    \caption{Q-Q plot for Randomized Quantile Residuals}
    \label{fig:enter-label}
\end{figure}

Here, we summarize the strengths and the weaknesses of the FB regression model. For a small-sized dataset, the FB model can be useful and perform as well as or better than other zero-inflated models. Especially for left-skewed count data with excess zeros, the FB model performs better than other existing ZI models. This is because the FB distribution is more versatile than other discrete distributions. However, the FB model is not suitable for a large dataset. It also suffers from computational instability in estimating the parameters. In this section, we used the numerical optimization method  ``Nelder-Mead" in R to find MLEs, which approximates the gradient of the objective function by the finite-difference calculations. We used 0 for the initial value of the parameters, but it was found that the estimates of the parameters changed with different initial values, and a convergence issue arose with some initial values.

It is known that the MLE for discrete data and the logistic regression model is notorious for not being robust in the presence of outliers, and various robust estimators have been developed \cite{robust5, robust1, robust4, robust3, robust2, robust6}. 
 For future work, it would be meaningful to study and apply these robust estimators to the FB model to alleviate computational challenges. For example, one can use the minimum distance estimators under the Cramer-von Mises distance  \cite{robust4} or the robust M-estimator based on the probability integral transformation  \cite{robust100}, which provides reliable results even when the data contains outliers.


\section{Discussion and conclusion}
We developed a new generalized linear model for count data that has many zeros. Our model utilizes the fractional binomial distribution that can serve as an over-dispersed, zero-inflated alternative to the regular binomial distribution. We used two datasets from agriculture and public health to fit zero-inflated regression models and FB regression models. The results show that the FB model is as versatile as or more versatile than other existing zero-inflated models to incorporate zero-inflation in count data. Especially for left-skewed count data, the FB model is found to be more appropriate than other models, as the FB model shows a good fit while other models show a lack of fit for such data.  

Although the FB regression model can serve as an additional tool for excess zero count data, its applicability is limited to a small size dataset with a few covariates (approximately fewer than 5) and a small or moderate bound on the response variable (approximately less than 20).
The FB model is computationally expensive and not suitable for a large data set since the pmf of the fractional binomial distribution does not have a closed-form expression, unlike Poisson and negative binomial distributions. The FB model also suffers from instability in numerical optimization for MLE. These drawbacks call for a different approach to parameter estimation and a scalable FB model for future work. 
\section*{Conflict of interest}
 The authors declare that they have no conflict of interest.
\section*{Data Sets and Computer Code}
R-code used in simulations and data analysis is available from GitHub repository \url{https://github.com/leejeo25/frbinom_regression.git} and
\url{https://github.com/leejeo25/fbglm.git}. 
Datasets used in this work are publicly available.

\normalem
\bibliographystyle{abbrv}
\bibliography{main}
\end{document}